\documentclass[onecolumn,english,twocolumn, final]{IEEEtran}
\usepackage[T1]{fontenc}
\usepackage[latin9]{inputenc}
\usepackage{xcolor}
\usepackage{cprotect}
\usepackage{mathrsfs}
\usepackage{bm}
\usepackage{amsmath}
\usepackage{amsthm}
\usepackage{amssymb}
\usepackage{graphicx}
\PassOptionsToPackage{normalem}{ulem}
\usepackage{ulem}

\makeatletter

\providecolor{lyxadded}{rgb}{0,0,1}
\providecolor{lyxdeleted}{rgb}{1,0,0}
\DeclareRobustCommand{\mklyxadded}[1]{\textcolor{lyxadded}\bgroup#1\egroup}
\DeclareRobustCommand{\mklyxdeleted}[1]{\textcolor{lyxdeleted}\bgroup\mklyxsout{#1}\egroup}
\DeclareRobustCommand{\mklyxsout}[1]{\ifx\\#1\else\sout{#1}\fi}
\DeclareRobustCommand{\lyxdeleted}[4][]{\mklyxdeleted{#4}}

\theoremstyle{plain}
\newtheorem{lem}{\protect\lemmaname}

\usepackage{amsmath}
\usepackage{amssymb}

\usepackage[level=3]{wgroup_message}  % set level to 0, to hide all notes
\usepackage{graphicx,psfrag,cite,subfigure}
\author{
Hao~Sun, Junting~Chen,~\IEEEmembership{Member,~IEEE}, and Xianghao~Yu,~\IEEEmembership{Senior Member,~IEEE}
\thanks{Hao Sun was with the School of Science and Engineering (SSE), Shenzhen Future Network of Intelligence Institute (FNii-Shenzhen), and
Guangdong Provincial Key Laboratory of Future Networks of Intelligence, The Chinese University of Hong Kong, Shenzhen, Guangdong 518172, China. He is now with the Department of Electrical Engineering, City University of Hong Kong, Hong Kong (e-mail:
hao.sun@cityu.edu.hk).}
\thanks{Junting Chen is with the School of Science and Engineering (SSE), Shenzhen Future Network of Intelligence Institute (FNii-Shenzhen), and
Guangdong Provincial Key Laboratory of Future Networks of Intelligence, The Chinese University of Hong Kong, Shenzhen, Guangdong 518172, China
(e-mail: juntingc@cuhk.edu.cn).}
\thanks{Xianghao Yu is with the Department of Electrical Engineering, City University of Hong Kong, Hong Kong (e-mail:alex.yu@cityu.edu.hk).}
}


\usepackage[acronym]{glossaries}
\newcommand{\newac}{\newacronym}
\newcommand{\ac}{\gls}

\makeglossaries

\newac{speb}{SPEB}{square position error bound}
\newac[plural=EFIMs,firstplural=Fisher information matrices (EFIMs)]{efim}{EFIM}{Fisher information matrix}
\newac{ne}{NE}{Nash equilibrium}
\newac{mse}{MSE}{mean squared error}
\newac{toa}{TOA}{time of arrival}
\newac{tdoa}{TDOA}{time difference of arrival}
\newac{snr}{SNR}{signal-to-noise ratio}
\newac{lan}{LAN}{local area network}
\newac{psd}{PSD}{positive semidefinite}
\newac{pd}{PD}{positive definite}
\newac{wrt}{w.r.t.}{with respect to}
\newac{lhs}{L.H.S.}{left hand side}
\newac{wp1}{w.p.1}{with probability 1}
\newac{kkt}{KKT}{Karush-Kuhn-Tucker}
\newac{wlog}{w.l.o.g.}{without loss of generality}
\newac{mle}{MLE}{maximum likelihood estimation}
\newac{rssi}{RSSI}{received signal strength indicator}
\newac{mimo}{MIMO}{multiple-input multiple-output}
\newac{csi}{CSI}{channel state information}
\newac{fdd}{FDD}{frequency division duplexing}
\newac{ms}{MS}{mobile station}
\newac{bs}{BS}{base station}
\newac{bss}{BSs}{base stations}
\newac{d2d}{D2D}{device-to-device}
\newac{slnr}{SLNR}{signal-to-interference-leakage-and-noise-ratio}
\newac{ula}{ULA}{uniform linear antenna array}
\newac{pas}{PAS}{power angular spectrum}
\newac{mmse}{MMSE}{minimum mean square error}
\newac{zf}{ZF}{zero-forcing}
\newac{rzf}{RZF}{regularized zero-forcing}
\newac{as}{AS}{angular spread}
\newac{ue}{UE}{user equipment}
\newac{ues}{UEs}{user equipments}
\newac{iid}{i.i.d.}{independent and identically distributed} 
\newac{sinr}{SINR}{signal-to-interference-and-noise ratio}
\newac{tdd}{TDD}{time-division duplex}
\newac{rvq}{RVQ}{random vector quantization}
\newac{rhs}{R.H.S.}{right hand side}
\newac{mrc}{MRC}{maximum ratio combining}
\newac{cdf}{CDF}{cumulative distribution function}
\newac{a.s.}{a.s.}{almost surely}
\newac{los}{LOS}{line-of-sight}
\newac{jsdm}{JSDM}{joint spatial division and multiplexing}
\newac{map}{MAP}{maximum a posteriori}
\newac{klt}{KLT}{Karhunen-Lo\`eve Transform}
\newac{lbe}{LBE}{link bargaining equilibrium}
\newac{se}{SE}{Stackelberg equilibrium}
\newac{uav}{UAV}{unmanned aerial vehicle}
\newac{uavs}{UAVs}{unmanned aerial vehicles}
\newac{nlos}{NLOS}{non-line-of-sight}
\newac{pdf}{PDF}{probability density function}
\newac{em}{EM}{expectation-maximization}
\newac{knn}{KNN}{$k$-nearest neighbors}
\newac{svd}{SVD}{singular value decomposition}
\newac{nmf}{NMF}{non-negative matrix factorization}
\newac{umf}{UMF}{unimodality-constrained matrix factorization}
\newac{rmse}{RMSE}{rooted mean squared error}
\newac{olos}{OLOS}{obstructed line-of-sight}
\newac{mmw}{mmW}{millimeter wave}
\newac{ber}{BER}{bit error rate}
\newac{rss}{RSS}{received signal strength}
\newac{lp}{LP}{linear program}
\newac{ufw}{U-FW}{unimodal Frank-Wolfe}
\newac{utf}{UTF}{unimodality-constrained tensor factorization}
\newac{fw}{FW}{Frank-Wolfe}
\newac{iot}{IoT}{Internet-of-Things}
\newac{mae}{MAE}{mean absolute error}
\newac{crb}{CRB}{Cram\'er-Rao bound}
\newac{aoa}{AoA}{angle of arrival}
\newac{wcl}{WCL}{weighted centroid localization}
\newac{slf}{SLF}{spatial loss function}
\newac{btd}{BTD}{block-term tensor decomposition}
\newac{tps}{TPS}{thin plate spline}
\newac{nmse}{NMSE}{normalized mean squared error}
\newac{svt}{SVT}{singular value thresholding}
\newac{vae}{VAE}{variational autoencoder}
\newac{gan}{GAN}{generative adversarial network}
\newac{aod}{AOD}{angle of departure}
\newac{rbf}{RBF}{radial basis function}
\newac{loocv}{LOOCV}{leave one out cross validation}
\newac{lpr}{LPR}{local polynomial regression}
\newac{xl-mimo}{XL-MIMO}{extremely large multiple-input multiple-output}
\newac{6g}{6G}{sixth-generation}
\newac{cpd}{CPD}{canonical polyadic decomposition}
\newac{nnm}{NNM}{nuclear norm minimization}
\newac{mimo-ofdm}{MIMO-OFDM}{multiple-input multiple-output orthogonal frequency division multiplexing}
\newac{mimo-cdma}{MIMO-CDMA}{multiple-input multiple-output code division multiple access}
\newac{pad}{PAD}{power-angular-delay}
\newac{hmm}{HMM}{hidden Markov model}
\newac{dft}{DFT}{discrete time Fourier transform}
\newac{pnp}{PnP}{Plug-and-Play}
\newac{ode}{ODE}{ordinary differential equation}
\newac{gps}{GPs}{Gaussian processes}
\newac{ai}{AI}{artificial intelligence}
\newac{2d}{2D}{two-dimensional}
\newac{ddpm}{DDPM}{denoising diffusion probabilistic modeling}

\setkeys{Gin}{width=1.0\columnwidth}

\renewcommand{\lyxdeleted}[3]{{\color{lyxdeleted}{}}}
\makeatother

\usepackage{babel}
\providecommand{\lemmaname}{Lemma}

\begin{document}
\title{MIMO Beam Map Reconstruction via Toeplitz-Structured Matrix-Vector
Tensor Decomposition}
\maketitle
\begin{abstract}
As wireless networks progress toward sixth-generation (6G), understanding
the spatial distribution of directional beam coverage becomes increasingly
important for beam management and link optimization. Multiple-input
multiple-output (MIMO) beam map provides such spatial awareness, yet
accurate construction under sparse measurements remains difficult
due to incomplete spatial coverage and strong angular variations.
This paper presents a tensor decomposition approach for reconstructing
MIMO beam map from limited measurements. By transforming measurements
from a Cartesian coordinate system into a polar coordinate system,
we uncover a matrix-vector outer-product structure associated with
different propagation conditions. Specifically, we mathematically
demonstrate that the matrix factor, representing beam-space gain,
exhibits an intrinsic Toeplitz structure due to the shift-invariant
nature of array responses, and the vector factor captures distance-dependent
attenuation. Leveraging these structural priors, we formulate a regularized
tensor decomposition problem to jointly reconstruct line-of-sight
(LOS), reflection, and obstruction propagation conditions. Simulation
results confirm that the proposed method significantly enhances data
efficiency, achieving a normalized mean square error (NMSE) reduction
of over $20$\% compared to state-of-the-art baselines, even under
sparse sampling regimes. 
\end{abstract}


\section{Introduction}

Massive \ac{mimo} has become a key enabler for fifth-generation (5G)
and is expected to remain central in future sixth-generation (6G)
networks, owing to its ability to provide high spatial multiplexing
gains, strong beamforming gains, and flexible interference mitigation
\cite{XueJiCMaS:J24,BraBehSay:J13,YuXZhaHaeLet:J17}. However, fully
harvesting these benefits requires high-dimensional \ac{csi}, which
leads to substantial channel training and feedback overhead. To alleviate
this burden, the MIMO beam map has emerged as a structured representation
that describes the spatial distribution of directional beams, signal
strength, and link quality across the coverage area \cite{CheCheCui:J25,wang2025beamckmframeworkchannelknowledge}.
By enabling interference management, adaptive beamforming, and energy-efficient
power allocation without exhaustive channel training \cite{HaTRomLop:C24,ChoMicLovKro:J21},
MIMO beam map provides a powerful tool for planning and optimizing
dense urban and complex indoor deployments \cite{MalFanYan:C17,ChiMasAltJem:J24,FenLinWanWan:J24}.

Nevertheless, how to efficiently construct accurate MIMO beam map
from limited measurements remains largely unexplored. A first challenge
arises from the incomplete spatial coverage due to sparse measurements.
When observations are limited, it is hard to fully capture beam propagation
patterns or signal interactions, resulting in missing regions and
underrepresented areas in the map. Second, MIMO beam map typically
exhibits stronger spatial variations, as the \ac{rss} is influenced
not only by the propagation environment but also by the beamforming
pattern. Consequently, conventional interpolation- or tensor-based
methods generally require a larger number of measurements to achieve
satisfactory accuracy.

MIMO beam map can be viewed as a particular type of radio map \cite{ZenCheXuWu:J24}.
Classical interpolation methods, such as Kriging \cite{SatSutIna:J21,VerFunRaj:J16},
local polynomial regression \cite{Fan:b96}, and kernel-based schemes
\cite{HamBefBal:C17}, estimate unobserved locations by exploiting
spatial smoothness or neighborhood information. Kriging typically
models spatial correlation functions explicitly, which may incur substantial
computational overhead in large-scale scenarios and may require prior
knowledge of the underlying spatial structure. Local polynomial regression
and kernel-based approaches rely on locality and smoothness assumptions,
which may become less effective when dealing with high-dimensional
feature spaces or irregularly distributed measurements. Matrix-based
methods leverage the observation that radio maps, when arranged in
matrix form, often exhibit low-rank characteristics. This property
enables techniques such as Bayesian learning \cite{WanZhuLin:J24b,WanZhuLin:J24},
dictionary learning \cite{KimGia:C13} and matrix completion \cite{SunChe:J22}
to infer missing measurements or extrapolate beam-related information
in unobserved regions. Tensor-based methods \cite{ZhaFuWan:J20,CheWanHua:J25,ChoMicLovKro:J21}
further exploit multi-dimensional structures inherent in radio maps
by modeling correlations across multiple domains, such as space, frequency,
or beam indices, using low-rank tensor decomposition. However, these
methods may not fully utilize the underlying latent structure in MIMO
beam map, which arises from distinct beam patterns associated with
various angles. 

Deep learning-based approaches including \ac{vae} \cite{TegRom:J21},
\ac{gan} \cite{HuHuaChen:J23,TimShrFux:J24,ShrFuHong:J22}, and diffusion-based
methods \cite{WanTaoCheYin:J25,LuoLiZPenChe:J25}, treat the radio
map as one or more layers of 2D images, using neural networks to recognize
common patterns based on extensive training data. These approaches
require large amounts of training data, which may not always be available.

To exploit the underlying structure embedded in beam patterns across
different beam indices, this paper first transforms the sparse Cartesian
measurements into a polar coordinate system parameterized by the beam
angle at the \ac{bs}, the spatial angle of the \ac{ue}, and the
propagation distance. These measurements are then aggregated into
a three-dimensional (3D) tensor, where each entry corresponds to the
\ac{rss} value within a specific polar bin defined by the angle and
distance parameters. Under different propagation conditions (e.g.,
\ac{los}, reflection, and obstruction), the resulting tensor exhibits
a matrix-vector outer-product structure. The matrix factor, referred
to as the beam-space gain, captures the interaction between the beamforming
pattern and the spatial angles of the UE, whereas the vector factor
models distance-dependent power attenuation. Furthermore, the beam-space
gain matrix is shown to exhibit an inherent Toeplitz structure, reflecting
the shift-invariant angular correlation induced by angular sampling
and beam index ordering. 

Building on this structural insight, we develop a Toeplitz-structured
matrix-vector tensor decomposition framework to reconstruct the MIMO
beam map in the presence of multiple propagation conditions. In the
general case where \ac{los}, reflection, and obstruction components
coexist, the beam-space tensor is modeled as a sum of a few matrix-vector
terms, each corresponding to one propagation condition. The resulting
optimization problem is tackled via an alternating minimization procedure.
Specifically, the matrix factors are updated by proximal-gradient
iterations under the Toeplitz regularization, whereas the vector factors
are obtained from a constrained quadratic program. 

The contributions of this paper are summarized as follows:
\begin{itemize}
\item We propose a tensor representation based on polar coordinate transformation
for MIMO beam map construction. The tensor model takes the form of
matrix-vector outer-products, where the matrix component is shown
to exhibit a Toeplitz structure that enables efficient construction
under sparse samples.
\item We develop a generalized Toeplitz-structured tensor decomposition
formulation to handle complex scenarios with strong reflections and
obstructions. Specifically, we incorporate a soft Toeplitz regularization
for the beam-space matrices, coupled with monotonicity constraints
on the distance-domain attenuation vectors to ensure physical consistency.
\item We design a structure-aware alternating minimization for the proposed
matrix-vector tensor decomposition model. For general propagation
scenarios, the algorithm alternately updates the beam-space gain matrix
and the distance-domain attenuation vectors via a Toeplitz-like regularization.
For a special case, where only a pure LOS region exists, we further
exploit the symmetric Toeplitz structure and impose it as a hard constraint,
which allows each beam-space gain matrix to be parameterized by a
single Toeplitz row and estimated via a closed-form least-squares
update, reducing the degrees of freedom and computational complexity.
\item Extensive simulations across different sampling ratios and propagation
environments demonstrate that the proposed Toeplitz-structured tensor
decomposition method obtains more than $20$\% improvement in reconstruction
NMSE compared with KNN, TPS, and conventional BTD-based baselines.
\end{itemize}

The rest of the paper is organized as follows. Section~\ref{sec:System-Model}
introduces the propagation, beam, and measurement models. Section~\ref{sec:Coordinate-Transformation-for}
develops a structured tensor representation of the MIMO beam map by
transforming Cartesian measurements into the polar domain and revealing
the underlying Toeplitz matrix-vector structure. Section~\ref{sec:Toeplitz-Structured-Matrix-Vecto}
formulates the Toeplitz-regularized matrix-vector tensor decomposition
problem and presents an alternating minimization algorithm for its
solution. Section \ref{sec:Scenario-Specific-Simplified-Mod} derives
simplified decomposition models for a pure LOS region. Numerical results
are presented in Section \ref{sec:Numerical-Results} and conclusion
is given in Section \ref{sec:Conclusion}.

\emph{Notation:} Vectors are written as bold italic letters $\bm{x}$,
matrices as bold capital italic letters $\bm{X}$, and tensors as
bold calligraphic letters $\bm{\mathcal{\bm{X}}}$. For a matrix $\bm{X}$,
$[\bm{X}]_{ij}$ denotes the entry in the $i$th row and $j$th column
of $\bm{X}$. For a tensor $\bm{\mathcal{\bm{X}}}$, $[\bm{\mathcal{X}}]_{i,j,k}$
denotes the entry under the index $(i,j,k)$. The symbol `$\circ$'
represents the outer product, `$*$' represents the element-wise product,
$\|\cdot\|_{2}$ represents the $l_{2}$ norm, and $\|\cdot\|_{F}$
represents the Frobenius norm. The operator $\mathrm{vec}(\cdot)$
stacks the columns of a matrix (or the entries of a tensor slice)
into a column vector, and $\mathrm{diag}(\bm{x})$ denotes a diagonal
matrix with the entries of $\bm{x}$ on its main diagonal. The transpose
is denoted by $(\cdot)^{\text{T}}$. $\mathbb{E}\{\cdot\}$ denotes
statistical expectation. $\mathcal{O}(x)$ means $|\mathcal{O}(x)|/x\leq C$,
for all $x>x_{0}$ with $C$ and $x_{0}$ are positive real numbers. 

\section{System Model\label{sec:System-Model}}

\cprotect\subsection{Propagation and Beam Model
\begin{figure}
\protect\begin{centering}
\protect\includegraphics{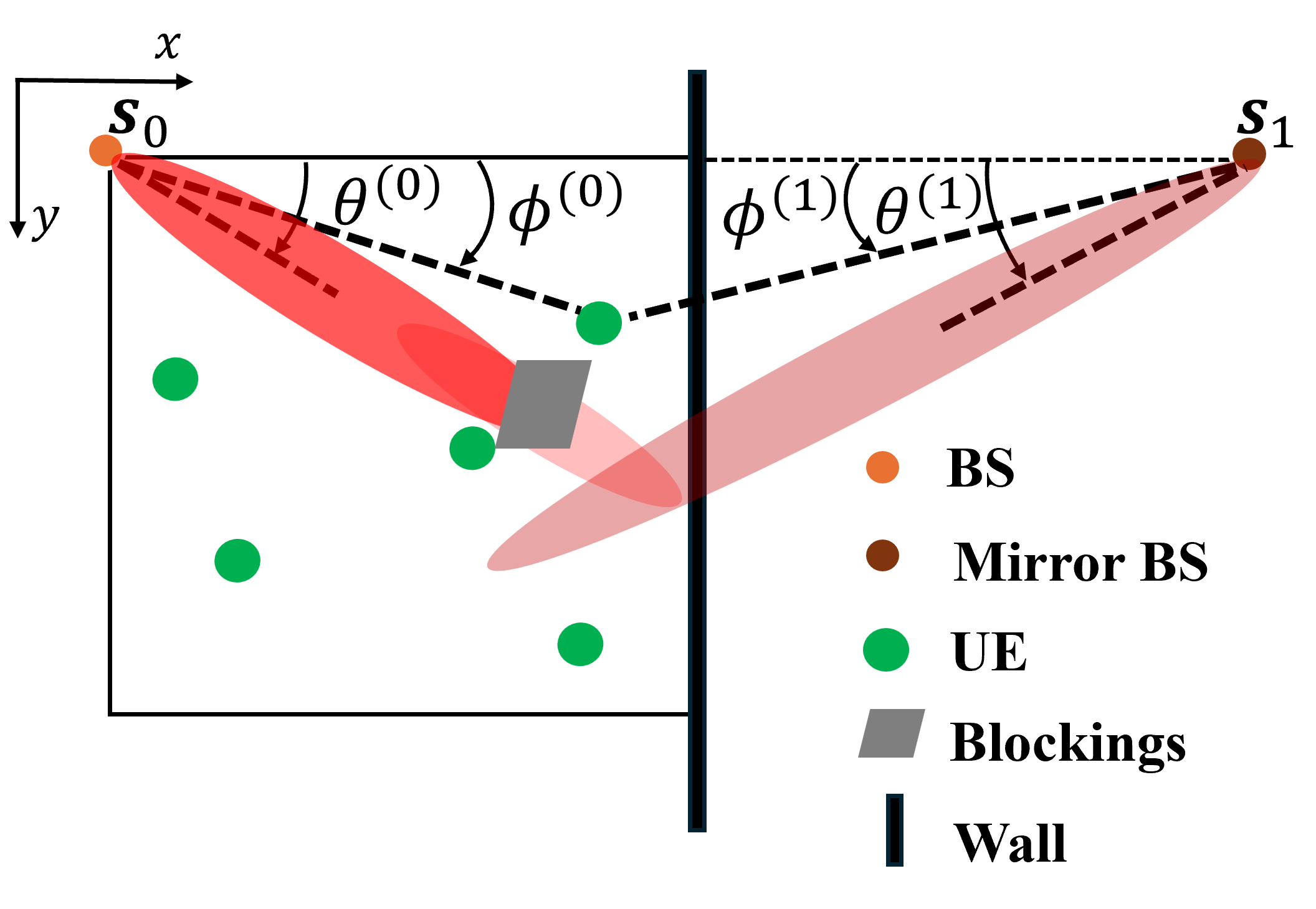}\protect
\par\end{centering}
\protect\caption{\label{fig:block_reflect}Complex scenario showing a direct path from
the BS, a reflected path via the mirror BS, and an obstructed path
affected by physical blockings.}
\end{figure}
}

Consider a downlink channel from the \ac{bs} to \ac{ues} in an indoor
environment. The \ac{bs}, located at the origin $\bm{s}_{0}=(0,0)$,
is equipped with a \ac{ula} of $N_{t}$ elements. The array elements
are aligned along the $y$-axis, with their pointing direction toward
the positive $x$-axis. The indoor region $\mathcal{D}$ of interest
is located in the first quadrant of the coordinate system as shown
in Fig.~\ref{fig:block_reflect}.

Consider the radio propagation from $\bm{s}_{0}$ to any position
in $\mathcal{D}$. Only a single reflection from the right wall in
Fig.~\ref{fig:block_reflect} is taken into account, since higher-order
reflections suffer severe attenuation due to multiple propagation
losses and thus contribute negligibly compared with the direct and
single-reflection paths \cite{GongZhaoHuKLuQ:J22}. Thus, the right
wall, which directly faces the BS, is the dominant reflector. According
to the law of specular reflection, we introduce a mirror position
$\bm{s}_{1}$, defined such that the right wall is the perpendicular
bisector of the line segment between $\bm{s}_{0}$ and $\bm{s}_{1}$.
In this way, the reflected path via the right wall is geometrically
equivalent to a mirror direct path from the mirror BS at $\bm{s}_{1}$.

The \ac{bs} is configured with a set of $I$ beams for environment
sensing, where each beam has a single mainlobe pointing at a specific
direction $\theta_{i}$, $i=1,2,\dots,I$, and the sidelobes have
significantly lower energy than that of the mainlobe. A typical example
of constructing such beams under \ac{ula} is to use the set of \ac{dft}
vectors for beamforming \cite{BraBehSay:J13,Say:J02}. As the \ac{bs}
is placed at the corner of $\mathcal{D}$, it suffices to only consider
the beam directions $\theta_{i}\in(0,\pi/2)$ for the sensing purpose.

Denote the array response vector $\bm{e}(\theta)\in\mathbb{C}^{N_{t}}$
of the \ac{ula} at \ac{bs} as
\[
\bm{e}(\theta)=[1,\ \text{exp}(-\jmath\pi\text{sin}\theta),\cdots,\text{exp}(-\jmath\pi(N_{t}-1)\text{sin}\theta)]^{\text{T}}
\]
where $\jmath=\sqrt{-1}$. The channel vector $\bm{h}$ for a position
$\bm{\bm{z}}\in\mathcal{D}$ is modeled as 
\begin{equation}
\bm{h}=\alpha_{0}\bm{e}(\phi_{0}(\bm{z}))e^{-\jmath2\pi f_{c}\tau_{0}}+\alpha_{1}\bm{e}(\phi_{1}(\bm{z}))e^{-\jmath2\pi f_{c}\tau_{1}}+\bm{h}^{\epsilon}\label{eq:channel_vec}
\end{equation}
where $\alpha_{0}$ and $\tau_{0}$ denote, respectively, the complex
gain and the delay of the direct path, which can possibly be significantly
attenuated due to blockage. The function $\phi_{0}(\bm{z})$ captures
the angle from the \ac{bs} at $\bm{s}_{0}$ to the position $\bm{z}$.
Likewise, $\alpha_{1}$ and $\tau_{1}$ denote the complex gain and
delay of the possible reflected path from the right wall. Note that,
such a path can also be blocked or may not exist, leading to $\alpha_{1}\approx0$.
The function $\phi_{1}(\bm{z})$ captures the angle from the mirror
BS to the position $\bm{z}$, representing the angle of the reflected
path leaving. Finally, the term $\bm{h}^{\epsilon}$ captures all
the residual paths.

It is possible that the reflecting environment is more complicated
and there are a few more mirror BSs. The construction of the reflection
due to additional mirror BSs is similar to the second term in (\ref{eq:channel_vec}).

\subsection{Measurement Model}

The direct and reflected paths can be separated by exploiting the
inherent characteristics of received signals \cite{GongZhaoHuKLuQ:J22,DuJCaoJinLiS:J24}.
Based on this, we establish the measurement model. For the $i$th
beam, the RSS associated with the direct path at location $\bm{z}_{m}$
is modeled as
\begin{equation}
\gamma_{i}^{(0)}(\bm{z}_{m})=\mathbb{E}\left\{ \left|\alpha_{0}\bm{w}_{i}^{\mathrm{H}}\bm{e}(\phi_{0}(\bm{z}_{m}))e^{-\jmath2\pi f_{c}\tau_{0}}\right|^{2}\right\} +n^{(0)}\label{eq:measurement_model_dir}
\end{equation}
where $\mathbb{E}\{\cdot\}$ denotes the expectation over the randomness
induced by small-scale fading contained in $\alpha_{0}$, $\bm{w}_{i}\triangleq\bm{e}(\theta_{i})$
is the beamforming vector, and $n^{(0)}$ captures the residual energy
contributed by scattered paths whose delays are close to that of the
direct path and thus cannot be fully resolved. We assume that $n^{(0)}$
follows a Gaussian distribution, i.e., $n^{(0)}\sim\mathcal{N}(0,\sigma_{0}^{2})$.

Similarly, the RSS corresponding to the reflected path is modeled
as
\begin{equation}
\gamma_{i}^{(1)}(\bm{z}_{m})=\mathbb{E}\left\{ \left|\alpha_{1}\bm{w}_{i}^{\mathrm{H}}\bm{e}(\phi_{1}(\bm{z}_{m}))e^{-\jmath2\pi f_{c}\tau_{1}}\right|^{2}\right\} +n^{(1)}\label{eq:measurement_model_ref}
\end{equation}
where $n^{(1)}$ accounts for the energy of scattered paths that exhibit
similar delays to the reflected path $\alpha_{1}$ and thus cannot
be distinguished. We also assume that $n^{(1)}\sim\mathcal{N}(0,\sigma_{1}^{2})$.

In the above formulation, it is assumed that the direct path $\alpha_{0}$
and the reflected path $\alpha_{1}$ can be extracted from the channel
vector $\bm{h}$ with some uncertainty, represented by $n^{(0)}$
and $n^{(1)}$, respectively. The accuracy of this extraction depends
on the resolution capability of the power-angular-delay profile, which
in turn is determined by the system bandwidth and antenna array structure.

Under this measurement model, the $i$th beam yields two types of
observations at location $\bm{z}_{m}$, namely $\gamma_{i}^{(0)}(\bm{z}_{m})$
and $\gamma_{i}^{(1)}(\bm{z}_{m})$, corresponding to the direct and
reflected paths, respectively. Collecting all such observations forms
two sparse measurement sets
\[
\bigl\{(\theta_{i},\bm{z}_{m},\gamma_{i}^{(0)}(\bm{z}_{m}))\bigr\}\quad\text{and}\quad\bigl\{(\theta_{i},\bm{z}_{m},\gamma_{i}^{(1)}(\bm{z}_{m}))\bigr\}.
\]
Our goal is to reconstruct the complete MIMO beam map from these sparse
measurements.

\section{Toeplitz-Structured Tensor Model via Coordinate Transformation\label{sec:Coordinate-Transformation-for}}

In this section, we show that by transforming the measurement set
$\{\theta_{i},\bm{z}_{m},\gamma_{i}^{(l)}(\bm{z}_{m})\}$, originally
defined in the Cartesian coordinate system $(x,y)$, into the polar
coordinate system $(\phi,d)$, the resulting tensor exhibits a matrix-vector
structure, where under certain conditions, the matrix component becomes
\emph{Toeplitz}.

\subsection{Coordinate Transformation}

We focus on the \ac{los} scenario for illustrative purposes, as the
reflected-path case can be handled in a similar way with the assistance
from the mirror BS $s_{1}$. For notational simplicity, we omit the
superscript `$(l)$' in this subsection.

\subsubsection{Transformation from Cartesian to Polar Coordinates\label{subsec:From-Cartesian-Coordinate}}

Recall that the BS is located at $\bm{s}_{0}=(x_{0},y_{0})$. For
each measured UE position $\bm{z}_{m}=(x_{m},y_{m})$, its polar coordinates
relative to $\bm{s}_{0}$ are defined as
\begin{equation}
d_{m}=\bigl\|\bm{z}_{m}-\bm{s}_{0}\bigr\|_{2},\quad\phi_{m}=\text{arctan}\frac{y_{m}-y_{0}}{x_{m}-x_{0}}
\end{equation}
where $\text{arctan}$ returns the polar angle of the point $(x_{m},y_{m})$
\ac{wrt} the positive $x$-axis.

To construct the measurement tensor, we discretize the angular and
radial domains into uniform grids 
\[
\{\phi_{1},\ldots,\phi_{J}\},\quad\phi_{j}=(j-1)\Delta_{\phi}
\]
\[
\{d_{1},\ldots,d_{K}\},\quad d_{k}=(k-1)\Delta_{d}
\]
where $\Delta_{\phi}$ and $\Delta_{d}$ denote the discretization
steps in the angular and radial domains. Each measurement $(\phi_{m},d_{m})$
is assigned to its nearest angular and radial bins $(\phi_{j},d_{k})$
via
\begin{equation}
j=1+\left\lfloor \frac{\phi_{m}}{\Delta_{\phi}}\right\rceil ,\ k=1+\left\lfloor \frac{d_{m}}{\Delta_{d}}\right\rceil ,
\end{equation}
where $\lfloor\cdot\rceil$ denotes rounding to the nearest integer.

Denote the measurement $\gamma_{i}(\bm{z}_{m})$ as $\gamma_{m,i}$.
We then construct a structured tensor $\bm{\mathcal{X}}\in\mathbb{R}^{I\times J\times K}$
where the entry $[\bm{\mathcal{X}}]_{i,j,k}$ aggregates all measurements
taken by beam $\theta_{i}$ from the polar bin $(\phi_{j},d_{k})$:
\begin{equation}
[\bm{\mathcal{X}}]_{i,j,k}=\begin{cases}
{\displaystyle \frac{1}{|\mathcal{S}_{j,k}|}\sum\limits_{m\in\mathcal{S}_{j,k}}\gamma_{m,i},} & \text{if }\mathcal{S}_{j,k}\neq\varnothing\\[1ex]
\text{unobserved}, & \text{otherwise}
\end{cases}\label{eq:aggregation}
\end{equation}
where $\mathcal{S}_{j,k}$ denotes the index set of the measurements
whose angular and radial coordinates are in the bin $(\phi_{j},d_{k})$.

\subsubsection{Toeplitz Structure}

To understand the intrinsic structure of the tensor $\bm{\mathcal{X}}$,
we analyze the measurement model for the direct path established in
(\ref{eq:measurement_model_dir}). The first term in (\ref{eq:measurement_model_dir})
represents the dominant signal component that we aim to model. For
a location $\bm{z}_{m}$ with polar coordinates $(\phi_{j},d_{k})$,
the RSS in (\ref{eq:measurement_model_dir}) satisfies
\begin{equation}
\mathbb{E}\left\{ \gamma_{m,i}\right\} =\mathbb{E}\left\{ |\alpha_{0}|^{2}\right\} \cdot\left|\bm{e}^{\mathrm{H}}(\theta_{i})\bm{e}(\phi_{j})\right|^{2}.\label{eq:separability}
\end{equation}
Then, (\ref{eq:measurement_model_dir}) can be written as a product
of a distance-dependent power attenuation term $P(d_{k})$ and a beamforming
gain term $\bm{e}^{\mathrm{H}}(\theta_{i})\bm{e}(\phi_{j})$
\begin{equation}
\gamma_{m,i}=P(d_{k})\cdot\left|\bm{e}^{\mathrm{H}}(\theta_{i})\bm{e}(\phi_{j})\right|^{2}+\epsilon\label{eq:redefine_r}
\end{equation}
where $P(d_{k})$ approximates $\mathbb{E}\{|\alpha_{0}|^{2}\}$ which
decays with distance $d_{k}$ and $\epsilon\sim\mathcal{N}(0,\delta^{2})$
captures the random noise that is not captured by the first term.

Based on the structure of (\ref{eq:redefine_r}), we define a distance-domain
attenuation vector $\bm{\rho}\in\mathbb{R}^{K}$, whose entries are
given by $\rho_{k}=P(d_{k})$. We further define a beam-space gain
matrix $\bm{G}\in\mathbb{R}^{I\times J}$, which exclusively captures
the interaction between the beamforming angles $\theta_{i}$ and the
spatial angles $\phi_{j}$. Its $(i,j)$-th entry is given by
\begin{align}
[\bm{G}]_{i,j} & =\left|\bm{e}^{\mathrm{H}}(\theta_{i})\bm{e}(\phi_{j})\right|^{2}\nonumber \\
 & =\left|\sum_{n=0}^{N_{t}-1}\exp\left(\jmath\pi n(\sin(\theta_{i})-\sin(\phi_{j}))\right)\right|^{2}.\label{eq:gain_matrix}
\end{align}

Consequently, from (\ref{eq:redefine_r})\textendash (\ref{eq:gain_matrix}),
the measurements $\gamma_{m,i}$ form a tensor $\bm{\mathcal{X}}$
that can be represented by a rank-1 matrix-vector outer product model
\begin{equation}
\bm{\mathcal{X}}=\bm{G}\circ\bm{\rho}+\bm{\mathcal{E}},\label{eq:tensor_model}
\end{equation}
where $\bm{\mathcal{E}}$ captures the residual energy which is modeled
as noise.

A crucial property of the beam-space gain matrix $\bm{G}$ is that,
under specific conditions, it becomes highly structured, i.e., Toeplitz
and symmetric.
\begin{lem}[Toeplitz and Symmetric Structure]
\label{lem:symemtric}If the angles $\theta_{i}$ and $\phi_{j}$
are chosen such that their sines form uniform grids, i.e., $\sin(\theta_{i})=i/I$
and $\sin(\phi_{j})=j/J$, and $I=J$, then the matrix $\bm{G}$ exhibits
a Toeplitz and symmetric structure.
\end{lem}
\begin{proof}
See Appendix \ref{sec:Proof-of-Lemma symmetric}.
\end{proof}
This structural property is of particular importance because it implies
that the entire beam-space gain matrix $\bm{G}$ only has a small
degrees-of-freedom, thereby facilitating efficient tensor reconstruction
from sparse measurements.

\begin{figure}
\subfigure[]{\includegraphics[width=0.5\columnwidth]{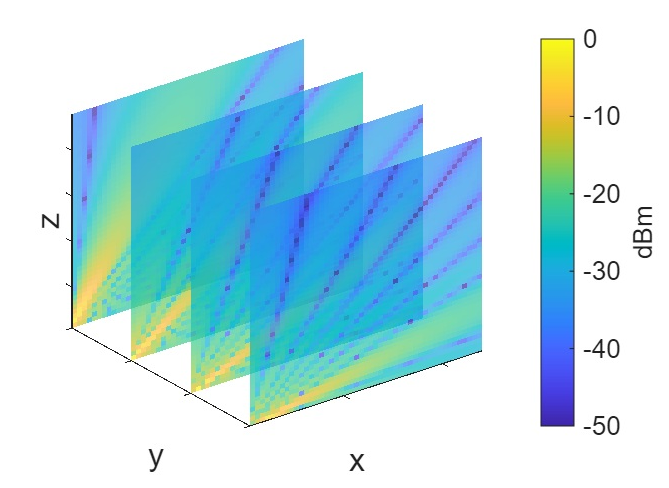}}\subfigure[]{\includegraphics[width=0.5\columnwidth]{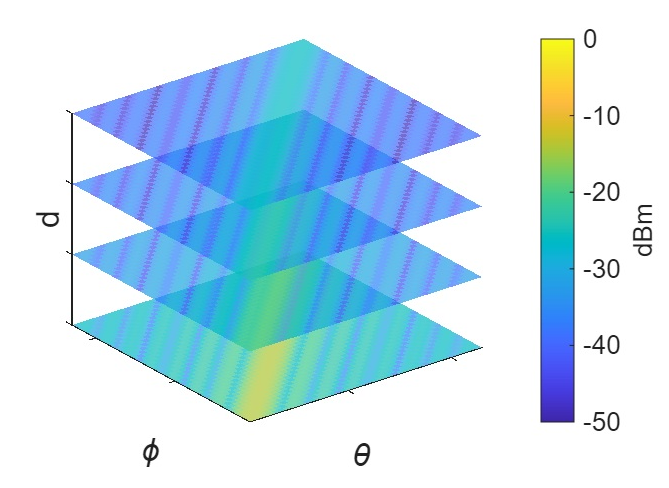}}\caption{\label{fig:The-original-beam and transformed}(a) The MIMO beam map
under the Cartesian coordinate system. (b) The transformed MIMO beam
map under the polar coordinate system, revealing Toeplitz and symmetric
structures that can be exploited for structured reconstruction.}
\end{figure}
Fig. \ref{fig:The-original-beam and transformed}(a) shows the MIMO
beam map $\bm{\mathcal{X}}$ represented in the Cartesian coordinate
system, where the beam pattern across different directions can be
directly visualized. However, in this form, the map does not explicitly
exhibit exploitable structural properties. By contrast, Fig. \ref{fig:The-original-beam and transformed}(b)
presents the direct beam component of $\bm{\mathcal{X}}$ in the polar
coordinate system, in which clear Toeplitz and symmetry structures
emerge. These structural characteristics are critical for enabling
reliable reconstruction of the MIMO beam map from sparse measurements
and are not readily observable in the Cartesian representation.
\begin{figure*}
\subfigure[]{\includegraphics[width=0.33\textwidth]{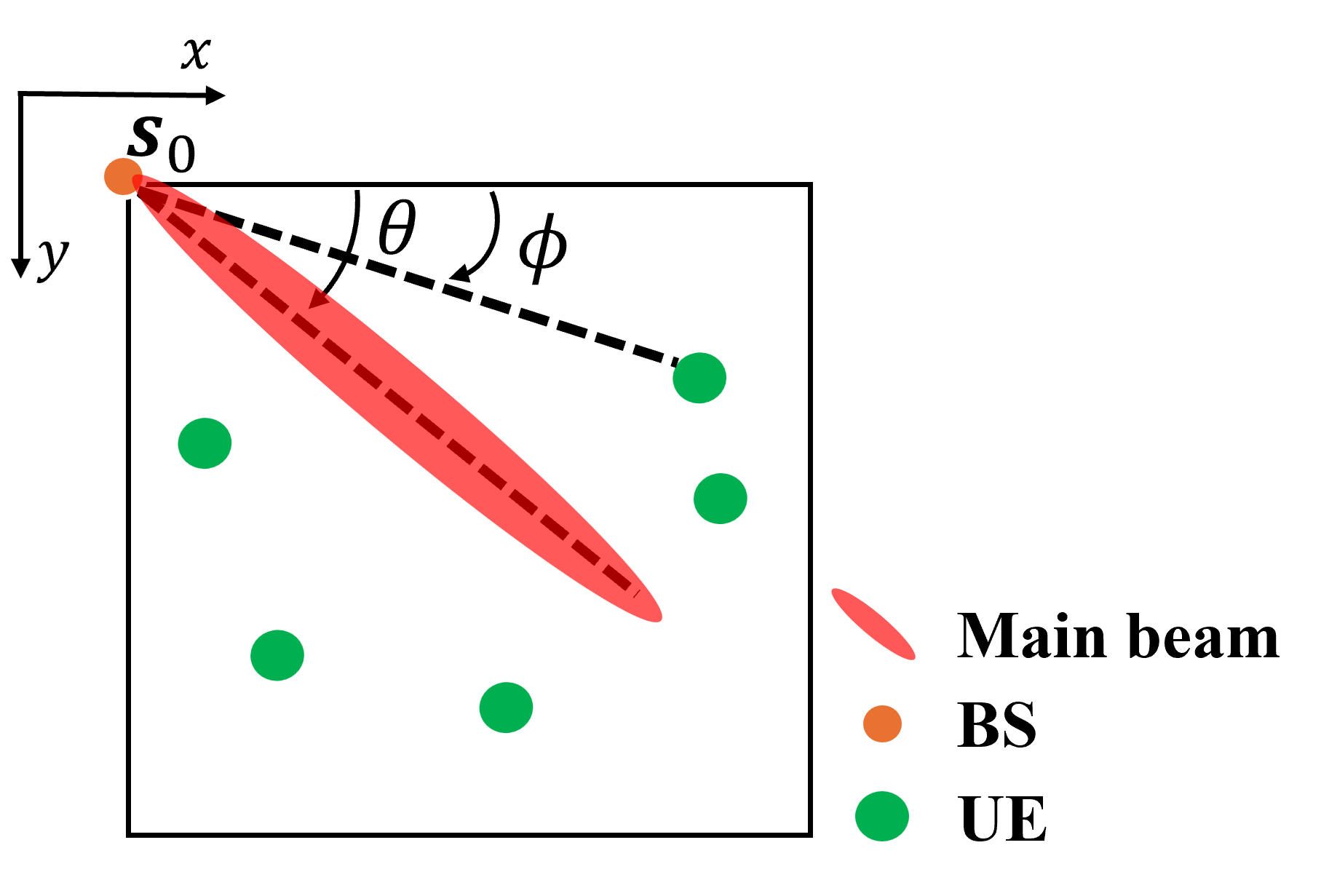}}
\subfigure[]{\includegraphics[width=0.33\textwidth]{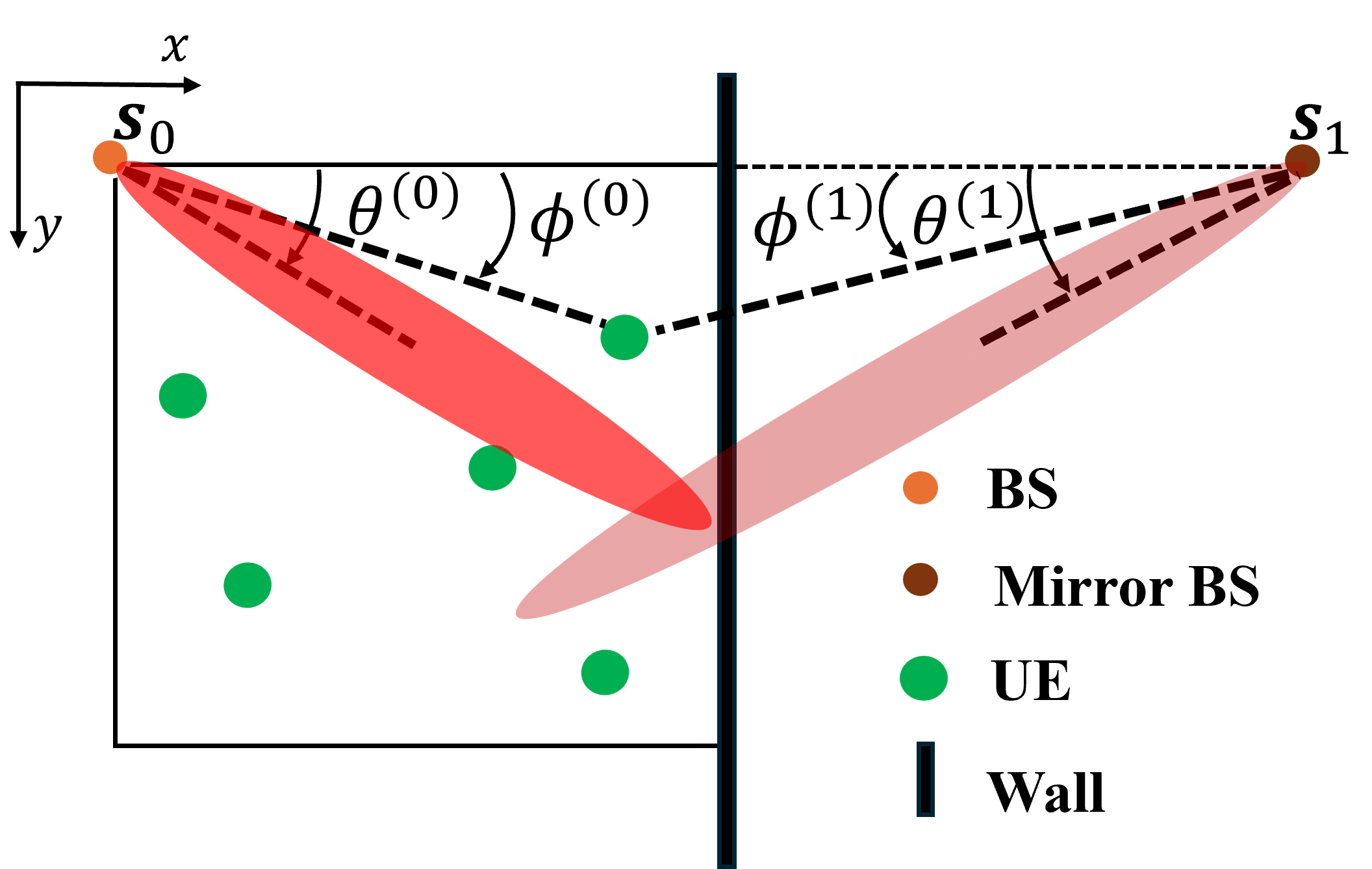}}\subfigure[]{\includegraphics[width=0.33\textwidth]{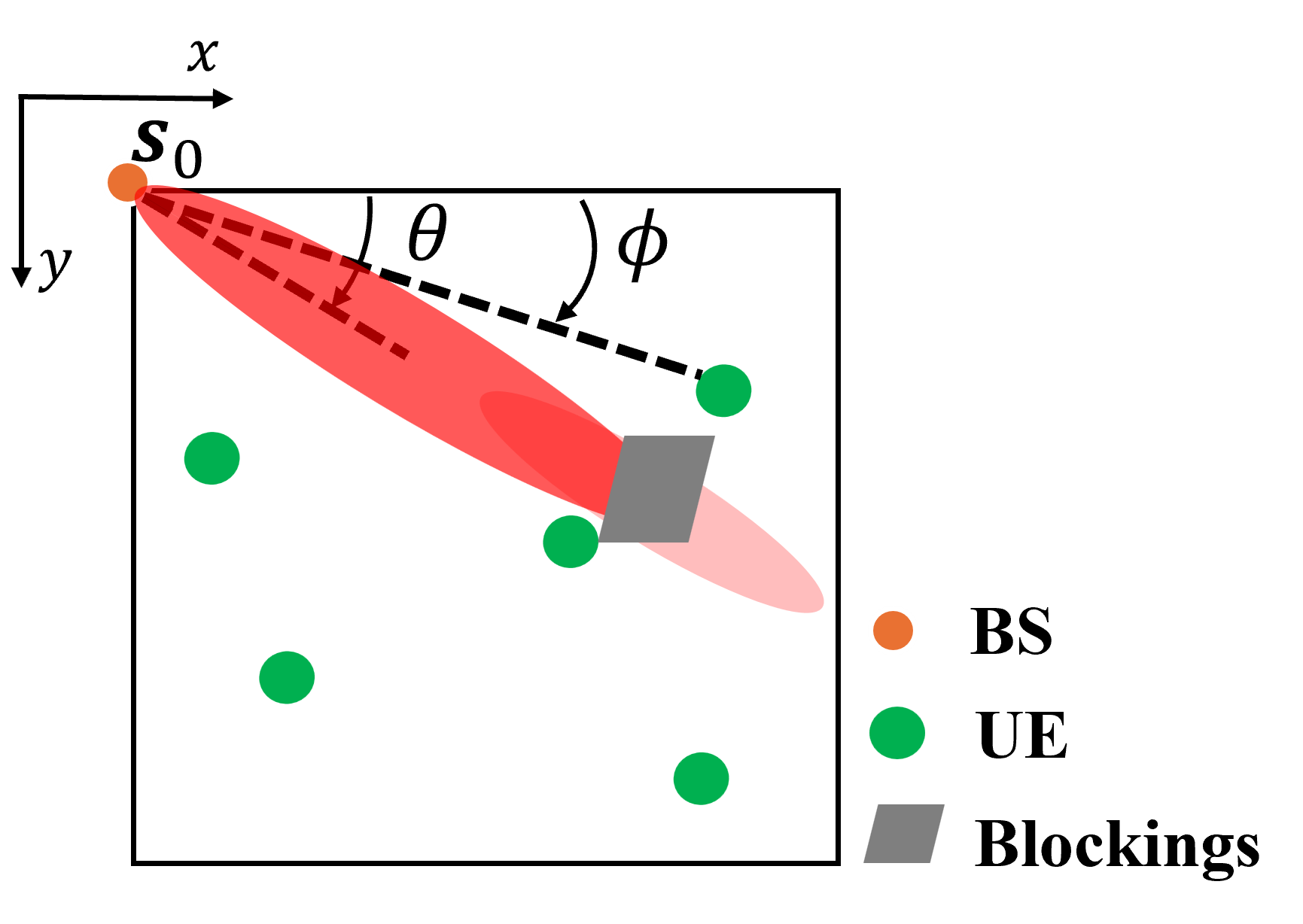}}

\caption{\label{fig:scenario(a)(b)(c)}Three common beam propagation scenarios:
(a) Pure LOS Region. (b) Coexistence of LOS and Reflection. (c) Coexistence
of LOS and Obstruction.}
\end{figure*}

\subsection{Tensor Model under the Presence of Reflection and Obstruction\label{subsec:Tensor-Model}}

With the presence of reflection and obstruction, we propose a rank-$R$
superposition tensor model to capture the dominant effect of the propagation
as
\begin{equation}
\bm{\mathcal{X}}=\sum_{r=1}^{R}\bm{G}_{r}\circ\bm{\rho}_{r}+\mathcal{\bm{E}}\label{eq:general_model}
\end{equation}
where each matrix $\bm{G}_{r}\in\mathbb{R}^{I\times J}$ represents
the beam-space gain pattern associated with the $r$-th propagation
condition, and $\bm{\rho}_{r}\in\mathbb{R}^{K}$ models its corresponding
distance-dependent attenuation. The rank $R$ reflects the number
of propagation conditions in the environment, typically characterized
by the coexistence of LOS, reflection, and obstruction.

The general model (\ref{eq:general_model}) can be adapted to capture
various propagation scenarios, as outlined in the following subsections.

\subsubsection{Pure LOS Region\label{subsec:Direct}}

In the pure LOS region where there is no strong reflection, as illustrated
in Fig.~\ref{fig:scenario(a)(b)(c)}(a), the tensor $\bm{\mathcal{X}}$
can be represented by a simple matrix-vector tensor decomposition
as in (\ref{eq:tensor_model}), where the rank in (\ref{eq:general_model})
reduces to $R=1$.

\subsubsection{Coexistence of LOS and Reflection\label{subsec:direct_reflect}}

For a region where LOS and reflection coexist, as illustrated in Fig.~\ref{fig:scenario(a)(b)(c)}(b),
similar to (\ref{eq:separability})\textendash (\ref{eq:redefine_r}),
the measurements corresponding to each propagation condition can be
written as
\[
\gamma_{m,i}^{(l)}=P(d_{k}^{(l)})\cdot\left|\bm{e}^{\mathrm{H}}(\theta_{i})\bm{e}(\phi_{j}^{(l)})\right|^{2}+\epsilon
\]
where $l=0$ denotes the LOS condition and $l=1$ denotes the reflection
condition. For each propagation condition, the measurement tensor
$\bm{\mathcal{X}}$ can be represented by a simple matrix-vector tensor
decomposition
\[
\bm{\mathcal{X}}^{(l)}=\bm{G}^{(l)}\circ\bm{\rho}^{(l)}+\mathcal{\bm{E}},\quad l=0,1.
\]

Since the reflection originates from a mirror BS $\bm{s}_{1}$, which
has an identical beam pattern as the BS $\bm{s}_{0}$, under the same
scope of $\phi_{j}^{(0)}$ and $\phi_{j}^{(1)}$, we have $\bm{G}^{(0)}=\bm{G}^{(1)}\triangleq\bm{G}$.
By concatenating the separated measurements along the distance dimension,
an augmented tensor can be constructed as
\begin{equation}
\bm{\mathcal{X}}=\text{cat}(\bm{\mathcal{X}}^{(0)},\bm{\mathcal{X}}^{(1)})+\mathcal{\bm{E}}=\bm{G}\circ\bm{\rho}+\mathcal{\bm{E}}\label{eq:direct_reflect}
\end{equation}
where $\bm{\rho}=[\bm{\rho}^{(0)};\bm{\rho}^{(1)}]\in\mathbb{R}^{2K}$.
This construction preserves the rank-1 structure of the tensor.

\subsubsection{Coexistence of LOS and Obstruction\label{subsec:direct_obstructed}}

In the case where the direct path is partially obstructed, as illustrated
in Fig.~\ref{fig:scenario(a)(b)(c)}(c), two distinct conditions
exist: a LOS region before the obstruction and an NLOS region after
it.

Consequently, the formulation in (\ref{eq:redefine_r}) is modified
as
\begin{equation}
\gamma_{m,i}=P_{0}(d_{k},\phi_{j})\cdot\left|\bm{e}^{\mathrm{H}}(\theta_{i})\bm{e}(\phi_{j})\right|^{2}+\epsilon\label{eq:redefine_r-1}
\end{equation}
where the power attenuation term $P_{0}$ now depends jointly on both
$d_{k}$ and $\phi_{j}$. For a fixed $d_{k}$, different $\phi_{j}$
may correspond to either LOS or NLOS regions, resulting in a lower
received power $\gamma_{m,i}$ when blockage exists.

To explicitly capture the angular dependence, we further express (\ref{eq:redefine_r-1})
as
\begin{equation}
\gamma_{m,i}=P(d_{k})\cdot\alpha(\phi)\left|\bm{e}^{\mathrm{H}}(\theta_{i})\bm{e}(\phi_{j})\right|^{2}+\epsilon\label{eq:redefine_r-1-1}
\end{equation}
where $\alpha(\phi)$ is a directional attenuation factor that varies
with $\phi$. Accordingly, for distance $d_{k}$ where both LOS and
NLOS regions coexist, we approximate the corresponding measurements
$\gamma_{m,i}$ using $\bm{G}_{2}\circ\bm{\rho}_{2}$, while the purely
LOS region is represented by $\bm{G}_{1}\circ\bm{\rho}_{1}$. This
leads to a rank-$R=2$ model in which the overall tensor is expressed
as a sum of two rank-one components
\begin{equation}
\bm{\mathcal{X}}=\bm{G}_{1}\circ\bm{\rho}_{1}+\bm{G}_{2}\circ\bm{\rho}_{2}+\mathcal{\bm{E}}.\label{eq:direct_block}
\end{equation}

A visual description of $\bm{G}_{1}$ and $\bm{G}_{2}$ is shown in
Fig.~\ref{fig:Visible-plot-of_G12}. The matrix $\bm{G}_{1}$ exhibits
a clear Toeplitz structure, whereas $\bm{G}_{2}$ shows a quasi-Toeplitz
pattern due to the presence of obstruction, where the darker region
corresponds to the obstructed region.

\begin{figure}
\subfigure{\includegraphics[width=0.5\columnwidth]{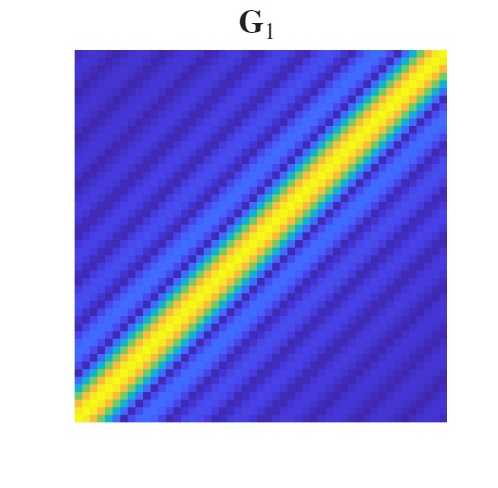}}\subfigure{\includegraphics[width=0.5\columnwidth]{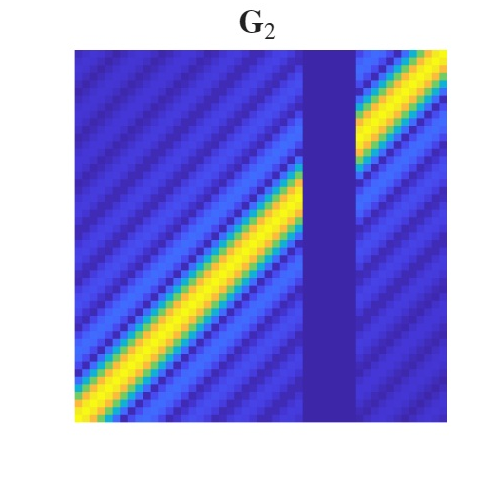}}\caption{\label{fig:Visible-plot-of_G12}Visible plots of $\bm{G}_{1}$ and
$\bm{G}_{2}$. $\bm{G}_{1}$ exhibits a Toeplitz structure corresponding
to the LOS region, while $\bm{G}_{2}$ displays a quasi-Toeplitz pattern
with a dark region indicating signal obstruction.}
\end{figure}

\subsubsection{Coexistence of LOS, Reflection, and Obstruction}

The scenario in Fig.~\ref{fig:block_reflect} involves the simultaneous
presence of LOS, reflection, and obstruction. In this case, the observed
beam pattern can be interpreted as the superposition of contributions
from the real BS $\bm{s}_{0}$ (without wall reflection) and its mirror
BS $\bm{s}_{1}$ (modeling the wall reflection).

When considering only the propagation originating from BS $\bm{s}_{0}$,
the setting is identical to that in Section \ref{subsec:direct_obstructed}.
Region $\bm{G}_{1}$ corresponds to the unobstructed LOS region, whereas
region $\bm{G}_{2}$ represents the obstruction region where the direct
path is blocked by the building.

For the propagation from the mirror BS $\bm{s}_{1}$, shown in Fig.~\ref{fig:scenario-blockage_reflection},
the LOS, reflection, and obstruction effects coexist. The two regions
$\bm{G}_{1}$ and $\bm{G}_{2}$ appear again, corresponding respectively
to the mirror-LOS area with distance to $\bm{s}_{1}$ less than $d_{1}$,
and the mirror-obstruction area with distance to $\bm{s}_{1}$ between
$d_{1}$ and $d_{2}$. In addition, a distinct region $\bm{G}_{3}$
emerges, where the reflected path is simultaneously blocked by both
the real building and its mirror image.

Following the same modeling principle as in Section \ref{subsec:direct_reflect},
this scenario can be represented by (\ref{eq:general_model}) with
$R=3$. 

Since the LOS and obstruction regions share the same spatial partitions
$\bm{G}_{1}$ and $\bm{G}_{2}$ as their mirror counterparts, the
corresponding attenuation components can be aggregated as
\begin{equation}
\bm{\rho}_{1}=\bm{\rho}_{1}^{(0)}+\bm{\rho}_{1}^{(1)},\qquad\bm{\rho}_{2}=\bm{\rho}_{2}^{(0)}+\bm{\rho}_{2}^{(1)}.
\end{equation}
In contrast, region $\bm{G}_{3}$ is affected solely by the mirror-induced
double obstruction and therefore contains only the corresponding attenuation
component $\bm{\rho}_{3}$.
\begin{figure}
\centering\includegraphics[width=0.8\columnwidth]{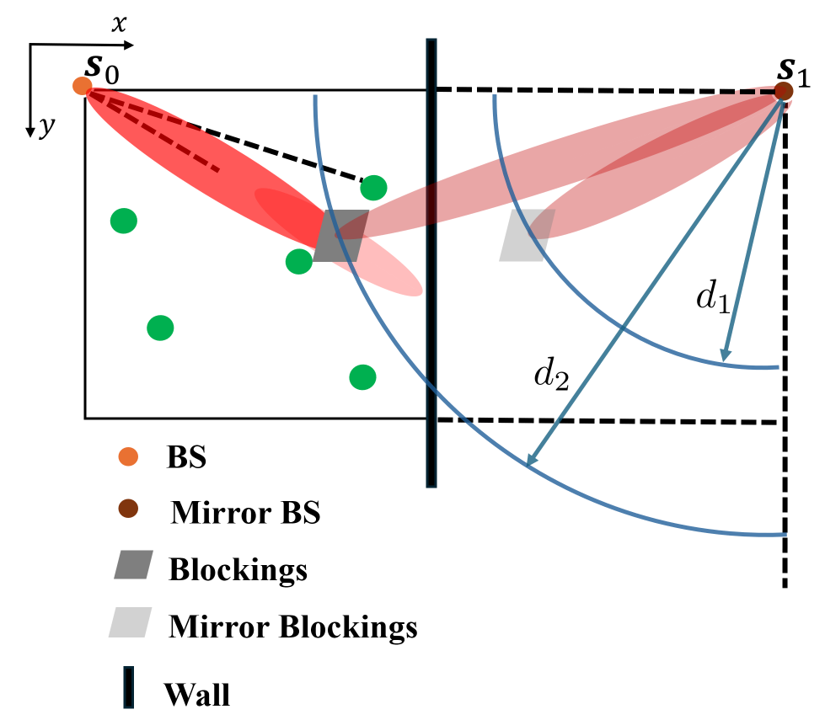}

\caption{\label{fig:scenario-blockage_reflection}Propagation from the mirror
BS $\bm{s}_{1}$ yields three distance regimes: (i) an LOS region
for $d\le d_{1}$; (ii) a single-obstruction region for $d_{1}<d\le d_{2}$;
and (iii) a double-obstruction region for $d>d_{2}$.}
\end{figure}

\section{Toeplitz-Structured Matrix-Vector Tensor Decomposition\label{sec:Toeplitz-Structured-Matrix-Vecto}}

In this section, we propose a matrix-vector tensor decomposition-based
optimization formulation to reconstruct the general scenario where
LOS, reflections, and obstructions coexist, as illustrated in Fig.~\ref{fig:block_reflect}.
The relatively simpler propagation scenarios discussed in Sections
\ref{subsec:Direct}\textendash \ref{subsec:direct_obstructed} can
be addressed in a similar manner as special cases of the proposed
formulation.

\subsection{Matrix-Vector Tensor Decomposition Formulation}

Building upon the coordinate transformation, the \ac{rss} measurements
in the polar domain can be expressed using a multi-term matrix-vector
tensor decomposition as given in (\ref{eq:general_model}). This representation
allows the reconstruction task to be formulated as a constrained optimization
problem, where the objective is to estimate the beam-space matrices
$\{\bm{G}_{r}\}$ and distance-domain attenuation vector $\{\bm{\rho}_{r}\}$
from the observed measurements while maintaining the underlying physical
characteristics.

To obtain a physically consistent and structurally regularized representation,
we exploit the near-Toeplitz structure of the beam-space gain matrices
$\bm{G}_{r}$. Rather than enforcing a strict Toeplitz constraint,
we incorporate a soft regularization term, since in environments where
both reflection and obstruction are present, the beam patterns only
exhibit approximate shift-invariance, as discussed in Section \ref{subsec:direct_obstructed}.
Moreover, because signal power decays with propagation distance, the
nonzero portion of each attenuation profile $\bm{\rho}_{r}$ is assumed
to be monotonically decreasing. These physical considerations directly
motivate the regularization terms and constraints adopted in the reconstruction
formulation.

Define $\bm{\mathcal{W}}\in\mathbb{R}^{I\times J\times K}$ as an
indication tensor such that
\[
[\bm{\mathcal{W}}]_{i,j,k}=\begin{cases}
1, & \text{if}\ (i,j,k)\in\Omega\\
0, & \text{otherwise}
\end{cases}
\]
where $\Omega$ denotes the index set of UE measurements in the polar
coordinate system. Then, the reconstruction of the MIMO beam map $\bm{\mathcal{X}}$
can be formulated as $\mathscr{P}$
\begin{align}
\mathscr{P}:\quad\underset{\{\bm{G}_{r}\},\{\bm{\rho}_{r}\}}{\text{minimize}} & \ \text{\ensuremath{\|}}\bm{\mathcal{W}}*(\bm{\mathcal{X}}-(\sum_{r=1}^{R}\bm{G}_{r}\circ\bm{\rho}_{r}))\|_{F}^{2}\label{eq:NLOS_BTD-1}\\
 & \quad+\sum_{r=1}^{R}\lambda_{r}\sum_{i=1}^{I-1}\sum_{j=1}^{J-1}|[\bm{G}_{r}]_{i,j}-[\bm{G}_{r}]_{i+1,j+1}|^{2}\label{eq:}\\
 & \qquad\qquad+\lambda\|[\bm{G}_{1},\cdots,\bm{G}_{R}]\|_{F}^{2}\\
\text{subject to} & \ \rho_{r,k}\geq0\ \forall k,r\label{eq:-1}\\
 & \ \sum_{r,k}\rho_{r,k}=1\label{eq:-2}\\
 & \ \sum_{r}\rho_{r,k}\geq\sum_{r}\rho_{r,k+1}\ \forall k=1,\cdots,K-1\label{eq:-3}\\
 & \ \rho_{1,k}\geq\rho_{1,k+1}\ \forall k=1,\cdots,K-1\label{eq:-4}\\
 & \ [\bm{G}_{r}]_{i,j}>0\ \forall i,j,r,\label{eq:-5}
\end{align}
where `$*$\textquoteright{} represents element-wise product. The
term $[\bm{G}_{r}]_{i,j}-[\bm{G}_{r}]_{i+1,j+1}$ serves as a Toeplitz
regularization, the Frobenius norm $\|[\bm{G}_{1},\cdots,\bm{G}_{R}]\|_{F}^{2}$
prevents overfitting by penalizing large magnitudes. Each power term
$\bm{\rho}_{r}$ is constrained to be nonnegative as in (\ref{eq:-1})
and normalized such that the total powers $\bm{\rho}_{r}$ satisfy
(\ref{eq:-2}) to ensure stable convergence. Since the index $k$
corresponds to the distance, the total power should decay with $k$
as in (\ref{eq:-3}). Moreover, (\ref{eq:-4}) guarantees that the
first component\textquoteright s power profile $\bm{\rho}_{1}$ decreases
smoothly along the distance dimension. The same monotonicity constraint
is not imposed on the other components $\bm{\rho}_{r}$, as their
profiles may contain zeros at small index $k$ due to the presence
of reflections and blockages.

To facilitate computation, let the matrix $\underbar{\ensuremath{\bm{X}}}\in\mathbb{R}^{(I\times J)\times K}$
and $\text{\underbar{\ensuremath{\bm{W}}}}\in\mathbb{R}^{(I\times J)\times K}$
represent the mode-3 unfolding of tensor $\bm{\mathcal{X}}$ and $\bm{\mathcal{W}}$,
respectively. The mode-3 unfolding is defined as
\[
\underbar{\ensuremath{\bm{X}}}=[\text{vec}(\bm{\mathcal{X}}_{:,:,1}),\ \text{vec}(\bm{\mathcal{X}}_{:,:,2}),\cdots,\text{vec}(\bm{\mathcal{X}}_{:,:,K})]
\]
where $\bm{\mathcal{X}}_{:,:,k}$ is the $k$th slice of the tensor
$\bm{\mathcal{X}}$. Additionally, define the matrix $\text{\underbar{\ensuremath{\bm{G}}}}=[\text{vec}(\bm{G}_{1}),\cdots,\text{vec}(\bm{G}_{R})]\in\mathbb{R}^{(I\times J)\times R}$
and $\bm{\varrho}=[\bm{\rho}_{1},\cdots,\bm{\rho}_{R}]\in\mathbb{R}^{K\times R}$.
Then, the problem $\mathscr{P}$ can be rewritten as
\begin{align}
\mathscr{P}':\quad\underset{\text{\underbar{\ensuremath{\bm{G}}}},\bm{\varrho}}{\text{minimize}} & \ \text{\ensuremath{\|}}\text{\ensuremath{\underbar{\ensuremath{\bm{W}}}}}*\text{\ensuremath{\underbar{\ensuremath{\bm{X}}}}}-\text{\ensuremath{\underbar{\ensuremath{\bm{W}}}}}*(\underbar{\ensuremath{\bm{G}}}\bm{\varrho}^{\text{T}})\|_{F}^{2}\label{eq:btd model-3-1-2}\\
 & \ \ +\sum_{r=1}^{R}\lambda_{r}\sum_{i=1}^{I-1}\sum_{j=1}^{J-1}|[\bm{G}_{r}]_{i,j}-[\bm{G}_{r}]_{i+1,j+1}|^{2}\nonumber \\
 & \qquad+\lambda\|\bm{\underbar{\ensuremath{\bm{G}}}}\|_{F}^{2}\nonumber \\
\text{subject to} & \quad\rho_{k\text{,}r}\geq0,\ \forall k=1,\cdots,K,r=1,\cdots,R\nonumber \\
 & \ \sum_{k,r}\rho_{k,r}=1\nonumber \\
 & \ \sum_{r}\rho_{k,r}\geq\sum_{r}\rho_{k+1,r}\ \forall k=1,\cdots,K-1\nonumber \\
 & \ \rho_{k,1}\geq\rho_{k+1,1}\ \forall k=1,\cdots,K-1.\nonumber 
\end{align}

\subsection{Alternating Updates under Structural Regularization\label{subsec:Alternating-Updates-under}}

The problem $\mathscr{P}'$ is separately convex with respect to $\underbar{\ensuremath{\bm{G}}}$
and $\bm{\varrho}$ when the other variable is fixed. We therefore
adopt an alternating minimization scheme to iteratively update $\underbar{\ensuremath{\bm{G}}}$
and $\bm{\varrho}$ until convergence.

\textbf{Update of} $\ensuremath{\underbar{\ensuremath{\bm{G}}}}$:
Given $\bm{\varrho}$, the subproblem with respect to $\underbar{\ensuremath{\bm{G}}}$
is formulated as
\begin{align}
\underset{\underbar{\ensuremath{\bm{G}}}}{\text{minimize}} & \ \text{\ensuremath{\|}}\text{\ensuremath{\underbar{\ensuremath{\bm{W}}}}}*\text{\ensuremath{\underbar{\ensuremath{\bm{X}}}}}-\text{\ensuremath{\underbar{\ensuremath{\bm{W}}}}}*(\underbar{\ensuremath{\bm{G}}}\bm{\varrho}^{\text{T}})\|_{F}^{2}\label{eq:btd model-3-1-1-1}\\
 & \qquad+\sum_{r=1}^{R}\lambda_{r}\sum_{i=1}^{I-1}\sum_{j=1}^{J-1}|[\bm{G}_{r}]_{i,j}-[\bm{G}_{r}]_{i+1,j+1}|^{2}\nonumber \\
 & \qquad\qquad\quad+\lambda\|\underbar{\ensuremath{\bm{G}}}\|_{F}^{2}.\nonumber 
\end{align}
This subproblem is convex since all the terms are convex. Therefore,
it can be efficiently solved using proximal gradient descent \cite{Ber:B97}.

\textbf{Update of} $\bm{\rho}$: Given $\underbar{\ensuremath{\bm{G}}}$,
the subproblem with respect to $\bm{\varrho}$ is given by
\begin{align}
\underset{\bm{\varrho}}{\text{minimize}} & \ \text{\ensuremath{\|}}\text{\ensuremath{\underbar{\ensuremath{\bm{W}}}}}*\text{\ensuremath{\underbar{\ensuremath{\bm{X}}}}}-\text{\ensuremath{\underbar{\ensuremath{\bm{W}}}}}*(\underbar{\ensuremath{\bm{G}}}\bm{\varrho}^{\text{T}})\|_{F}^{2}\label{eq:btd model-3-1-2}\\
\text{\text{subject to}} & \quad\rho_{k\text{,}r}\geq0,\ \forall k=1,\cdots,K,r=1,\cdots,R\nonumber \\
 & \ \sum_{k,r}\rho_{k,r}=1\nonumber \\
 & \ \sum_{r}\rho_{k,r}\geq\sum_{r}\rho_{k+1,r}\ \forall k=1,\cdots,K-1\nonumber \\
 & \ \rho_{k,1}\geq\rho_{k+1,1}\ \forall k=1,\cdots,K-1.\nonumber 
\end{align}
This subproblem is convex with respect to $\bm{\varrho}$. Since the
objective function is quadratic and all constraints are linear, it
constitutes a standard quadratic program that can be efficiently solved
using interior-point methods \cite{PotWri:J00}.

Finally, the reconstructed tensor is obtained as $\hat{\bm{\mathcal{X}}}=\sum_{r=1}^{R}\hat{\bm{G}}_{r}\circ\hat{\bm{\rho}}_{r}$.

\section{Simplified Model for the Pure LOS Region\label{sec:Scenario-Specific-Simplified-Mod}}

In this section, we show that in the pure LOS region, enforcing the
symmetric Toeplitz structure as hard constraints leads to a simplified
formulation of problem $\mathscr{P}$, with substantially reduced
computational complexity compared to the regularized formulation.

\subsection{Simplified Model}

For pure LOS region illustrated in Fig. \ref{fig:scenario(a)(b)(c)}(a),
the matrix-vector tensor decomposition problem $\mathscr{P}'$ can
be simplified and formulated as
\begin{align*}
\mathscr{P}'':\quad\underset{\bm{G},\bm{\rho}}{\text{minimize}} & \ \text{\ensuremath{\|}}\text{\ensuremath{\underbar{\ensuremath{\bm{W}}}}}*\text{\ensuremath{\underbar{\ensuremath{\bm{X}}}}}-\text{\ensuremath{\underbar{\ensuremath{\bm{W}}}}}*(\text{vec}(\bm{G})\bm{\rho}^{\text{T}})\|_{F}^{2}\\
 & \quad\quad+\lambda_{1}\sum_{i=1}^{I-1}\sum_{j=1}^{J-1}|[\bm{G}]_{i,j}-[\bm{G}]{}_{i+1,j+1}|^{2}\\
 & \quad\quad\quad+\lambda_{2}\sum_{i=1}^{I}\sum_{j=1}^{J}|[\bm{G}]_{i,j}-[\bm{G}]{}_{j,i}|^{2}\\
\text{subject to} & \ \rho_{k}>\rho_{k+1}>0,\ \forall k=1,\cdots,K-1\\
 & \ \sum_{k}\rho_{k}=1.
\end{align*}
 This formulation can be solved using a similar alternating optimization
framework as that adopted for the general problem $\mathscr{P}'$.

Since there is only a LOS region, the Toeplitz and symmetric properties
can be imposed as hard constraints rather than through regularization.
Under these structural constraints, the matrix $\bm{G}$ can be fully
characterized by a single row vector, thereby converting the reconstruction
task into a constrained least-squares problem. This representation
remains effective even under severe measurement sparsity.

By enforcing Toeplitz and symmetry constraints, the matrix-vector
tensor decomposition for the LOS region can be formulated as
\begin{align}
\mathscr{P}''':\underset{\bm{G},\bm{\rho}}{\text{minimize}} & \ \text{\ensuremath{\|}}\text{\ensuremath{\underbar{\ensuremath{\bm{W}}}}}*\text{\ensuremath{\underbar{\ensuremath{\bm{X}}}}}-\text{\ensuremath{\underbar{\ensuremath{\bm{W}}}}}*(\text{vec}(\bm{G})\bm{\rho}^{\text{T}})\|_{F}^{2}\label{eq:btd toeplitz constraint}\\
\text{subject to} & \ [\bm{G}]_{i,j}=[\bm{G}]_{j,i},\ [\bm{G}]_{i,j}=[\bm{G}]_{i+1,j+1}\nonumber \\
 & \ i+1\leq I,j+1\leq J\nonumber \\
 & \ \rho_{k}>\rho_{k+1}>0,\ \forall k=1,\cdots,K-1\nonumber \\
 & \ \sum_{k}\rho_{k}=1.\nonumber 
\end{align}

We employ an alternating minimization scheme to iteratively solve
for $\bm{G}$ and $\bm{\rho}$, as detailed below.

\textbf{Update of} $\bm{\rho}$: Given $\bm{G}$, we recall that
$\text{\underbar{\ensuremath{\bm{X}}}\ensuremath{\in\mathbb{R}^{(I\times J)\times K}}}$
and $\text{\underbar{\ensuremath{\bm{W}}}}\in\mathbb{R}^{(I\times J)\times K}$
are the mode-3 unfoldings of $\bm{\mathcal{X}}$ and $\bm{\mathcal{W}}$
respectively.

Then, the problem (\ref{eq:btd toeplitz constraint}) becomes
\begin{align}
\underset{\bm{\rho}}{\text{minimize}} & \ \text{\ensuremath{\|}}\text{\ensuremath{\underbar{\ensuremath{\bm{W}}}}}*\text{\ensuremath{\underbar{\ensuremath{\bm{X}}}}}-\text{\ensuremath{\underbar{\ensuremath{\bm{W}}}}}*(\text{vec}(\bm{G})\bm{\rho}^{\text{T}})\|_{F}^{2}\label{eq:btd toeplitz constraint rho}\\
\text{\text{subject to}} & \quad\bm{A\rho}>0,\ \bm{\rho}>0,\ \bm{1}^{\text{T}}\bm{\rho}=1,\nonumber 
\end{align}
where $\bm{1}\in\mathbb{R}^{K}$ is all-ones vector, and $\bm{A}\in\mathbb{R}^{(K-1)\times K}$
is a constraint matrix defined as
\[
\bm{A}=\left[\begin{array}{ccccc}
1 & -1 & 0 & \cdots & 0\\
0 & 1 & -1 & \cdots & 0\\
\vdots & \vdots & \vdots & \ddots & \vdots\\
0 & 0 & \cdots & 1 & -1
\end{array}\right].
\]
Using the identity $\text{\ensuremath{\underbar{\ensuremath{\bm{W}}}}}*(\text{vec}(\bm{G})\bm{\rho}^{\text{T}})=\text{\text{diag}}(\text{vec}(\bm{G}))\text{\ensuremath{\underbar{\ensuremath{\bm{W}}}}}\text{diag}(\bm{\rho})$,
the objective function in (\ref{eq:btd toeplitz constraint rho})
can be expanded as
\begin{align*}
 & \text{\ensuremath{\|}}\text{\ensuremath{\underbar{\ensuremath{\bm{W}}}}}*\text{\ensuremath{\underbar{\ensuremath{\bm{X}}}}}-\text{\text{diag}}(\text{vec}(\bm{G}))\text{\ensuremath{\underbar{\ensuremath{\bm{W}}}}}\text{diag}(\bm{\rho})\|_{F}^{2}\\
= & \sum_{i=1}^{I\times J}\sum_{k=1}^{K}([\text{\ensuremath{\underbar{\ensuremath{\bm{W}}}}}]_{i,k}[\text{\underbar{\ensuremath{\bm{X}}}}]_{i,k}-\text{vec}(\bm{G})_{i}[\text{\underbar{\ensuremath{\bm{W}}}]}_{i,k}\rho_{k})^{2}.
\end{align*}
This problem can be further reformulated into the standard quadratic
programming form
\begin{align}
\underset{\bm{\rho}}{\text{minimize}} & \ \bm{\rho}^{\text{T}}\bm{H}\bm{\rho}+\bm{c}^{\text{T}}\bm{\rho}\label{eq:btd toeplitz constraint rho-1-1}\\
\text{subject to} & \quad\bm{A\rho}>0,\ \bm{\rho}>0,\ \bm{1}^{\text{T}}\bm{\rho}=1,\nonumber 
\end{align}
where $\bm{H}$ is diagonal matrix with $[\bm{H}]_{k,k}=2\sum_{i=1}^{I\times J}(\text{vec}(\bm{G})_{i}[\text{\underbar{\ensuremath{\bm{W}}}}]_{i,k})^{2}$,
and $\bm{c}\in\mathbb{R}^{K}$ is given by $c_{k}=-2\sum_{i=1}^{I\times J}[\text{\ensuremath{\underbar{\ensuremath{\bm{W}}}}}]_{i,k}[\text{\underbar{\ensuremath{\bm{X}}}}]_{i,k}\text{vec}(\bm{G})_{i}$.

This quadratic programming problem with a quadratic objective function
and linear inequality constraints can be solved via interior-point
method \cite{PotWri:J00}.

\textbf{Update of} $\bm{G}$: Given $\bm{\text{\ensuremath{\rho}}}$,
problem (\ref{eq:btd toeplitz constraint}) becomes
\begin{align}
\underset{\bm{G}}{\text{minimize}} & \ \text{\ensuremath{\|}}\bm{\mathcal{W}}*(\bm{\mathcal{X}}-(\bm{G}\circ\bm{\rho}))\|_{F}^{2}\label{eq:btd model-5}\\
\text{\text{subject to}} & \quad[\bm{G}]_{i,j}=[\bm{G}]_{j,i},\ [\bm{G}]_{i,j}=[\bm{G}]_{i+1,j+1}\nonumber \\
 & \quad i+1\leq I,j+1\leq J.\nonumber 
\end{align}

Under these hard constraints, matrix $\bm{G}$ is a symmetric Toeplitz
matrix. Recall that we consider a square grid with $I=J\triangleq N$.
Matrix $\bm{G}$ is fully determined by its first row $\bm{g}=[g_{0},g_{1},\cdots,g_{N-1}]^{\text{T}}$
and $[\bm{G}]_{i,j}=g_{|i-j|}$. Let $\bm{T}\in\mathbb{R}^{N^{2}\times N}$
denote the fixed transformation matrix such that $\text{vec}(\bm{G})=\bm{T}\bm{g}\in\mathbb{R}^{N^{2}}$.
Substituting this relation into the unfolded objective function (\ref{eq:btd model-5})
gives
\begin{align}
\underset{\bm{g}}{\text{minimize}} & \ \text{\ensuremath{\|}}\text{\ensuremath{\underbar{\ensuremath{\bm{W}}}}}*\text{\ensuremath{\underbar{\ensuremath{\bm{X}}}}}-\text{\text{diag}}(\bm{T}\bm{g})\text{\ensuremath{\underbar{\ensuremath{\bm{W}}}}}\text{diag}(\bm{\rho})\|_{F}^{2}.\label{eq:btd model-6-1}
\end{align}
Define $\bm{y}=\text{vec}(\text{\ensuremath{\underbar{\ensuremath{\bm{W}}}}}*\text{\ensuremath{\underbar{\ensuremath{\bm{X}}}}})$,
and construct a matrix $\bm{D}\in\mathbb{R}^{(N^{2}\times K)\times N}$such
that $[\bm{D}]_{(i-1)\times K+k,j}=T_{ij}[\text{\ensuremath{\underbar{\ensuremath{\bm{W}}}}}]_{i,k}\rho_{k}$.
Then, (\ref{eq:btd model-6-1}) reduces to the standard least-squares
problem
\begin{equation}
\underset{\bm{g}}{\text{minimize}}\ \|\bm{y}-\bm{D}\bm{g}\|_{2}^{2}.\label{eq:least square}
\end{equation}
This is an unconstrained strictly convex problem. Setting the first-order
derivative to zero yields the closed-form solution
\begin{equation}
\hat{\bm{g}}=(\bm{D}^{\text{T}}\bm{D})^{-1}\bm{D}^{\text{T}}\bm{y}.\label{eq:update_g}
\end{equation}
The matrix $\bm{G}$ is then reconstructed as $[\bm{G}]_{i,j}=\hat{g}_{|i-j|}.$

Finally, the reconstructed MIMO beam map is obtained as $\bm{\mathcal{X}}=\hat{\bm{G}}\circ\hat{\bm{\bm{\rho}}}$.

The experimental comparison between the regularized and hard-constrained
formulations is presented in Fig.~\ref{fig:Comparasion-between-Toeplitz},
and discussed in detail in Section \ref{subsec:Comparasion-between-Toeplitz}.

\subsection{Complexity Analysis}

In this subsection, we compare the computational complexity of the
\emph{regularized} formulation in $\mathscr{P}''$ and the \emph{hard-constrained}
formulation in $\mathscr{P}'''$ for the pure LOS case.

Both $\mathscr{P}''$ and $\mathscr{P}'''$ are solved by alternating
minimization between vector $\bm{\rho}$ and matrix $\bm{G}$. Let
$T_{\mathrm{alt}}$ denote the number of outer alternating iterations.
In addition, let $T_{\mathrm{ip}}$ denote the number of Newton steps
used by the interior-point solver for (\ref{eq:btd toeplitz constraint rho-1-1}),
and let $T_{\mathrm{pg}}$ denote the number of proximal-gradient
iterations used to update $\bm{G}$ in the regularized method.

\subsubsection{Update of $\bm{\rho}$ (Common to $\mathscr{P}''$ and $\mathscr{P}'''$)}

Given $\bm{G}$, updating $\bm{\rho}$ yields a $K$-variable quadratic
program with linear constraints. The coefficients $\bm{H}$ and $\bm{c}$
in (\ref{eq:btd toeplitz constraint rho}) can be formed by accumulating
only over observed entries, which costs $\mathcal{O}(|\Omega|)$.
Solving the quadratic program via an interior-point method typically
costs $\mathcal{O}(T_{\mathrm{ip}}K^{3})$ \cite{PotWri:J00}. Hence,
the per-outer-iteration cost of the $\bm{\rho}$-update (\ref{eq:btd toeplitz constraint rho-1-1})
is $\mathcal{O}\!\left(|\Omega|+T_{\mathrm{ip}}K^{3}\right).$

\subsubsection{Update of $\bm{G}$}

In $\mathscr{P}''$, $\bm{G}\in\mathbb{R}^{N\times N}$ contains $N^{2}$
free variables and is estimated by minimizing a data-fitting term
plus non-smooth structural regularizers. As described in Section \ref{subsec:Alternating-Updates-under},
this subproblem can be solved by proximal-gradient descent. In each
proximal-gradient iteration, computing the gradient of the quadratic
data-fitting term can be implemented over the observed set, at a cost
of $\mathcal{O}(|\Omega|)$, while evaluating the proximal mapping
associated with the structural regularizers requires operations over
all $N^{2}$ entries (cost $\mathcal{O}(N^{2})$ per iteration). Therefore,
the per-outer-iteration complexity of the $\bm{G}$-update in $\mathscr{P}''$
is $\mathcal{O}\!\left(T_{\mathrm{pg}}(|\Omega|+N^{2})\right).$

In $\mathscr{P}'''$, symmetry and Toeplitz are imposed as hard constraints,
so $\bm{G}$ is fully determined by its first row $\bm{g}=[g_{0},g_{1},\ldots,g_{N-1}]^{\text{T}}\in\mathbb{R}^{N}$
with $[\bm{G}]_{ij}=g_{|i-j|}$, reducing the degrees of freedom from
$N^{2}$ to $N$. Given $\bm{\rho}$, the $\bm{g}$-update reduces
to a least-squares problem (\ref{eq:update_g}). Since the value at
any position $(i,j)$ depends solely on the index difference $|i-j|$,
each observation contributes to the estimation of a unique coefficient
$g_{|i-j|}$ without coupling with others. This structural orthogonality
renders the Gram matrix $\bm{D}^{\text{T}}\bm{D}$ diagonal, allowing
$\bm{g}$ to be computed efficiently by simple grouped accumulations
over $\Omega$. Thus, the complexity is $\mathcal{O}(|\Omega|+N).$

\subsubsection{Overall complexity}

Combining the above, the overall complexities are summarized as
\begin{align}
\text{\ensuremath{\mathscr{P}''}:}\quad & \mathcal{O}\!\left(T_{\mathrm{alt}}\Big(|\Omega|+T_{\mathrm{ip}}K^{3}+T_{\mathrm{pg}}(|\Omega|+N^{2})\Big)\right),\label{eq:complexity_reg}\\
\ensuremath{\mathscr{P}'''}:\quad & \mathcal{O}\!\left(T_{\mathrm{alt}}\Big(|\Omega|+T_{\mathrm{ip}}K^{3}+(|\Omega|+N)\Big)\right).\label{eq:complexity_hard}
\end{align}

In sparse-sampling regimes, the dominant difference stems from the
spatial update: $\mathscr{P}''$ optimizes over $N^{2}$ unknowns
and requires iterative proximal updates with per-iteration cost scaling
as $|\Omega|+N^{2}$, whereas $\mathscr{P}'''$ reduces the spatial
degrees of freedom to $N$ and admits a least-squares update that
scales primarily with the number of observations $|\Omega|$.

\section{Numerical Results\label{sec:Numerical-Results}}

In this section, we use the models in (\ref{eq:measurement_model_dir})
and (\ref{eq:measurement_model_ref}) to simulate the MIMO beam map
over an $L\times L$ area with $L=48$ meters. The BS is located at
$\bm{s}_{0}=(0,0)$, and we set the number of antennas to $N_{t}=16$.
The large-scale fading coefficients are chosen as $\alpha_{0}=d^{\eta}$
and $\alpha_{1}=d^{\eta}/2$ where $d=||\bm{z}-\bm{s}_{0}\|_{2}$.
The path-loss exponent is set to $\eta=-2$ for unobstructed propagation
and $\eta=-6$ when the path is blocked. The noise terms $n^{(0)}$
and $n^{(1)}$ are modeled with standard deviations $\sigma_{0}=3$dB
and $\sigma_{1}=1$dB. We choose the dimension of $\bm{\mathcal{X}}$
to be $I=J=46$ and $K=40$. The beam angles $\{\theta_{i}\}_{i=1}^{I}$
are generated such that $\sin\theta_{i}$ is uniformly spaced over
$[\sin\theta_{1},\,\sin\theta_{I}]$, where $\theta_{1}=0.0208$ and
$\theta_{I}=1.55$. The spatial angle $\phi_{j}$ is chosen to be
the same as $\theta_{i}$ to satisfy Lemma \ref{lem:symemtric}.

As illustrated in Fig. \ref{fig:block_reflect}, a wall is assumed
to exist along the right-hand side of the region, connecting the vertices
$(48,0)$ and $(48,48)$. The corresponding mirror BS for the reflected
path is therefore located at $\bm{s}_{1}=(96,0)$. UE locations are
drawn uniformly at random in the $L\times L$ area to obtain RSS $\gamma_{m,i}^{(0)}$
and $\gamma_{m,i}^{(1)}$.

We first transform the measurements $\{(\theta_{i},\bm{z}_{m},\gamma_{m,i})\}$
into the polar coordinate system $(\phi,d)$ via (\ref{eq:aggregation}).
The proposed Toeplitz-structured matrix-vector tensor decomposition
method is then applied to reconstruct the MIMO beam map $\bm{\mathcal{X}}$
under the polar coordinate system. After decomposition, $\bm{\mathcal{X}}$
needs to be mapped back to the Cartesian coordinate system to obtain
MIMO beam map $\bm{\mathcal{Y}}\in\mathbb{R}^{L\times L\times I}$
to convey direct beam information. Each entry $[\bm{\mathcal{X}}]_{i,j,k}$
indexed by $(\theta_{i},\phi_{j},d_{k})$, is converted to a corresponding
Cartesian location $\bm{z}_{m_{0}}$, where $m_{0}=1,\cdots,M_{0}$,
and $M_{0}=K\times J$. This is because under a fixed beam $\theta_{i}$,
the number of values in $\mathcal{X}_{i,:,:}$ is $K\times J$. Denote
$\gamma_{m_{0}}$ as a one-to-one mapping to $\mathcal{X}_{i,j,k}$
at $\bm{z}_{m_{0}}$. Then, the set of measurements becomes $\{(\theta_{i},\bm{z}_{m_{0}},\gamma_{m_{0}})\}$.
It is noted that under a fixed $\theta_{i}$, there is a beam-space
matrix $\mathcal{Y}_{i,:,:}\in\mathbb{R}^{L\times L}$. We propose
to interpolate such matrix $\mathcal{Y}_{i,:,:}$ for each $\theta_{i}$
based on the measurements $\{(\theta_{i},\bm{z}_{m_{0}},\gamma_{m_{0}})\}$
using \ac{tps}. The reconstruction performance is evaluated using
the \ac{nmse}, defined as $\|\bm{\mathcal{Y}}-\hat{\bm{\mathcal{Y}}}\|_{F}^{2}/\|\bm{\mathcal{Y}}\|_{F}^{2}$.

We compare the proposed Toeplitz-structured matrix-vector tensor decomposition
approach with the following baseline methods:
\begin{itemize}
\item Baseline 1: The \ac{knn} method, with $k=3$.
\item Baseline 2: The \ac{tps} \cite{JuaGonGeo:J11}, applied independently
for each beam $\theta_{i}$ to construct the corresponding 2D beam
map from the measurements $\{(\theta_{i},\bm{z}_{m},\gamma_{m,i})\}$.
\item Baseline 3: The traditional \ac{btd} \cite{ZhaFuWan:J20} method
without imposing Toeplitz structure.
\end{itemize}
Note that low-rank tensor completion \cite{LiuMusWon:J13} exhibits
significantly degraded performance under highly sparse observations
and is therefore excluded from comparison. Baselines 1 and 2 are purely
2D interpolation methods and reconstruct the MIMO beam map for each
beam angle $\theta_{i}$ independently, without exploiting any structural
correlations.

\subsection{Simulation for Pure LOS Scenario\label{subsec:Simulation-for-Direction}}

In this subsection, we show the simulation results under a LOS scenario
in Fig.~\ref{fig:scenario(a)(b)(c)}(a).

\subsubsection{Comparison between Toeplitz Constraint and Toeplitz Regularization\label{subsec:Comparasion-between-Toeplitz}}

We compare the \ac{nmse} performance of reconstructing $\bm{\mathcal{X}}$
using the Toeplitz constraint formulation in problem $\mathscr{P}'''$
and the Toeplitz regularization formulation in problem $\mathscr{P}''$.
Fig.~\ref{fig:Comparasion-between-Toeplitz} reports the NMSE under
different sampling ratios, where the Toeplitz constraint method consistently
achieves lower NMSE than the regularization-based counterpart across
all sampling rates. The improvement is more pronounced in the low-sampling
regime, indicating that directly enforcing the Toeplitz structure
is particularly beneficial when measurements are sparse.

To further evaluate robustness, Fig. \ref{fig:Comparasion-between-Toeplitz-1}
shows the NMSE under different noise levels from $0$ to $4$ dB.
As the noise increases, the NMSE of both methods degrades, but the
Toeplitz constraint method maintains a clear performance advantage
throughout the tested SNR range. Overall, these results suggest that
when the Toeplitz prior accurately matches the underlying signal structure
in the pure LOS case, hard structural enforcement yields more reliable
reconstructions than soft regularization, and remains consistently
superior under both sampling ratios variation and noise perturbation.

\begin{figure}
\includegraphics{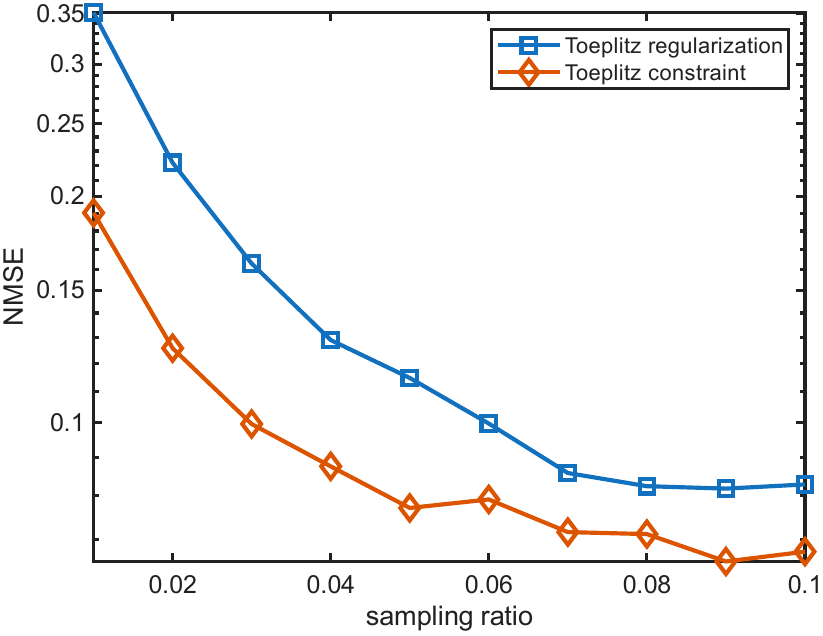}

\caption{\label{fig:Comparasion-between-Toeplitz}NMSE comparison between the
Toeplitz constraint and Toeplitz regularization methods under different
sampling ratios. The constraint method performs better than using
Toeplitz as regularization.}
\end{figure}

\begin{figure}
\includegraphics{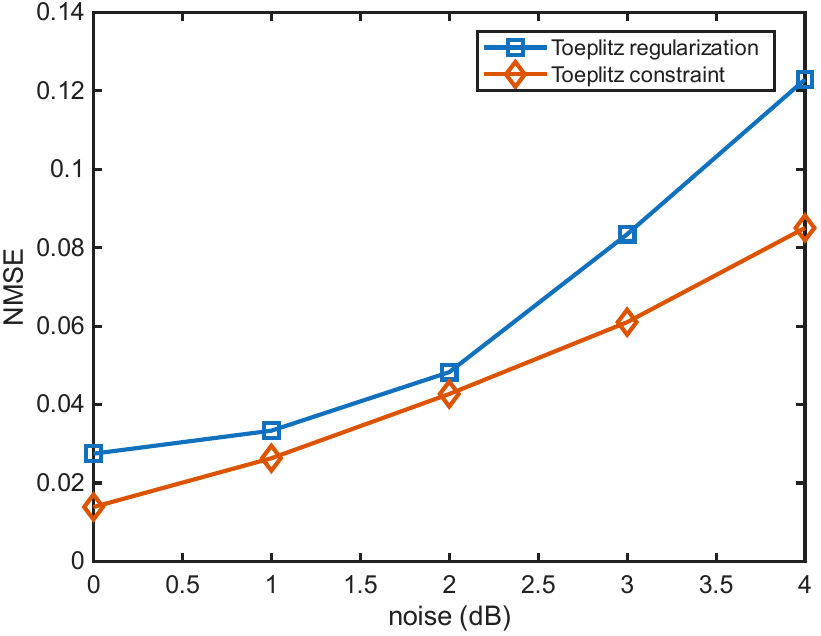}

\caption{\label{fig:Comparasion-between-Toeplitz-1}NMSE comparison between
the Toeplitz constraint and Toeplitz regularization methods under
different noise. The constraint method performs better than using
Toeplitz as regularization.}
\end{figure}

\subsubsection{The Performance under Different Sampling Ratios}

We evaluate the performance of the proposed method for reconstructing
$\bm{\mathcal{Y}}$ under sampling ratios $r_{s}=0.04-0.1$. The NMSE
results are presented in Fig.~\ref{fig:LOS sampling ratio }. The
proposed method consistently outperforms all baseline approaches,
achieving more than a $20$\% improvement at low sampling ratios.
The traditional BTD method fails to recover the MIMO beam map because
it does not exploit the Toeplitz structure. The TPS and KNN interpolation
methods also perform poorly, as they cannot capture the underlying
beam pattern. 
\begin{figure}
\includegraphics{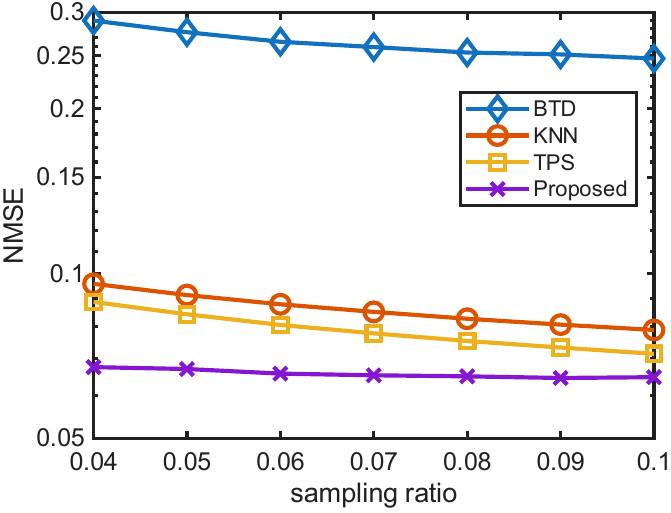}\caption{\label{fig:LOS sampling ratio }Reconstruction NMSE of $\bm{\mathcal{Y}}$
of LOS scenario under different sampling ratios}
\end{figure}

To further illustrate the reconstruction quality, we set the sampling
ratio to $r_{s}=0.1$ and show the visual results in Fig.~\ref{fig:Visual-plot-of direct}.
The proposed method accurately reconstructs the MIMO beam map by leveraging
the Toeplitz structure, while TPS and KNN can only provide coarse
approximations and fail to recover the fine beam features. 
\begin{figure}
\includegraphics{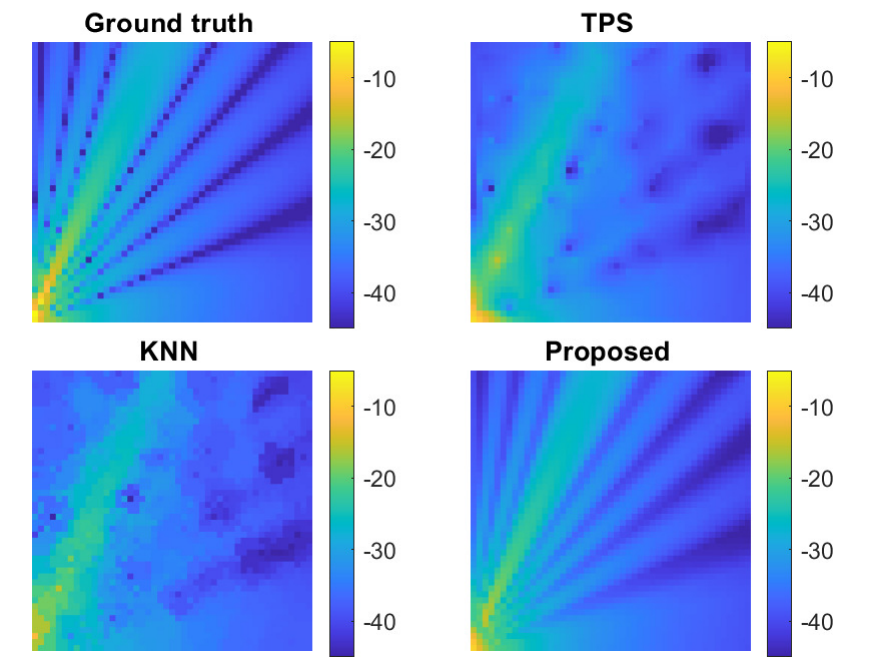}\caption{\label{fig:Visual-plot-of direct}Visual comparison of the reconstructed
direct beam map at sampling ratio $r_{s}=0.1$. The proposed method
successfully captures the beam pattern, while TPS and KNN produce
blurred and inaccurate reconstructions.}
\end{figure}

\subsection{Simulation for Coexistence of LOS and Reflection\label{subsec:direct_ref}}

In this subsection, we examine the reconstruction performance when
LOS and reflection coexist.

\subsubsection{Benefits of the Matrix-Vector Tensor Decomposition for Reflected-Beam
Reconstruction}

We demonstrate that the measurements from the direct-path beam can
improve the reconstruction of the reflected-path beam through the
proposed matrix-vector tensor formulation in (\ref{eq:direct_reflect}).
In Fig.~\ref{fig:The-BTD-tensor benefits relfect beam}, the method
\textquotedblleft Reflect only\textquotedblright{} uses only the reflection
measurements to reconstruct the MIMO beam map associated with the
mirror BS $\bm{s}_{1}$. In contrast, the method \textquotedblleft Reflect
with direct\textquotedblright{} jointly incorporates both the LOS
and reflection measurements using the tensor formulation introduced
in (\ref{eq:direct_reflect}).

As shown in Fig.~\ref{fig:The-BTD-tensor benefits relfect beam},
the joint formulation achieves more than a $10$\% improvement in
NMSE. This gain arises because the LOS and reflected beams share the
same intrinsic energy distribution $\bm{G}$. Therefore, the LOS measurements
provide additional information that helps estimate $\bm{G}$ more
accurately, which in turn benefits the reconstruction of the reflected
beam.

\begin{figure}
\includegraphics{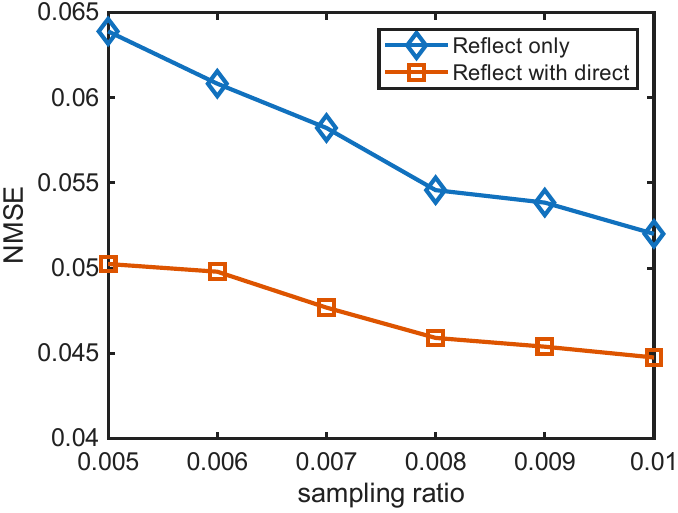}\caption{\label{fig:The-BTD-tensor benefits relfect beam}Reconstruction NMSE
comparison showing the benefit of incorporating direct-path measurements
when reconstructing the reflected beam using the matrix-vector tensor
formulation.}
\end{figure}

\subsubsection{The Performance under Different Sampling Ratios}

Here, similar to the LOS scenario, we evaluate the proposed method
under sampling ratios $r_{s}=0.04-0.1$ for reconstructing $\bm{\mathcal{Y}}$.
Since the direct-beam result is similar to the LOS case in Section
\ref{subsec:Simulation-for-Direction}, we focus on the reconstruction
NMSE of the whole MIMO beam map (the aggregate of direct and reflected
beams), as shown in Fig.~ \ref{fig:Reconstruction-NMSE sampling ratio}.
The proposed method achieves over a $20$\% improvement in NMSE compared
with all baseline methods, as illustrated in Fig. \ref{fig:Reconstruction-NMSE sampling ratio}. 

\begin{figure}
\includegraphics{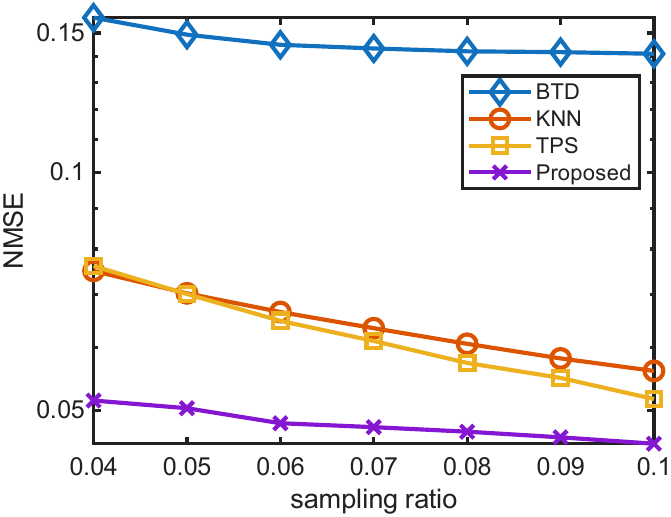}\caption{\label{fig:Reconstruction-NMSE sampling ratio}Reconstruction NMSE
of $\bm{\mathcal{Y}}$ for the coexistence of LOS and reflection scenario
under different sampling ratios. }
\end{figure}

We further set the sampling ratio to $r_{s}=0.1$ and visualize the
reconstruction results for the reflected beam and the whole MIMO beam
map in Fig.~\ref{fig:Visual-plot-of reflect} and Fig.~\ref{fig:Visual-plot-of whole},
respectively. The proposed method accurately reconstructs both components,
whereas TPS and KNN produce only coarse approximations and fail to
capture the underlying beam patterns.
\begin{figure}
\includegraphics{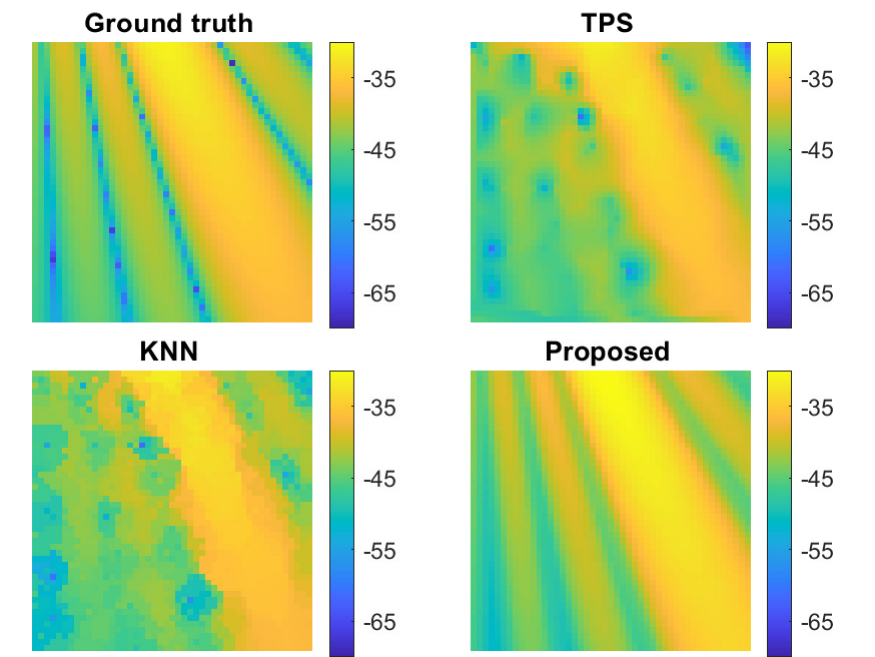}\caption{\label{fig:Visual-plot-of reflect}Visual comparison of the reconstructed
reflected beam map at sampling ratio $r_{s}=0.1$. The proposed method
successfully captures the reflected-beam pattern, while TPS and KNN
produce blurred and inaccurate reconstructions.}
\end{figure}
\begin{figure}
\includegraphics{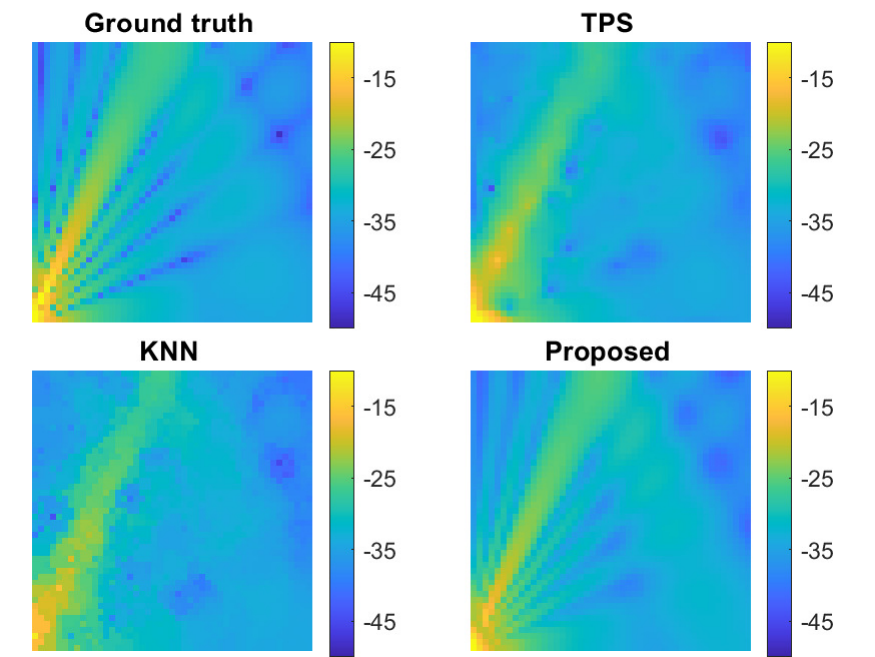}\caption{\label{fig:Visual-plot-of whole}Visual comparison of reconstructed
whole MIMO beam map (LOS + reflection) at sampling ratio $r_{s}=0.1$.
The proposed method reconstructs the beam structure faithfully, outperforming
TPS and KNN.}
\end{figure}

\subsection{Simulation for Coexistence of LOS and Obstruction}

In this subsection, we examine the reconstruction performance when
the direct path is obstructed by a building located within the area
of interest, as illustrated in Fig.~\ref{fig:scenario(a)(b)(c)}(c).
The polygonal coordinates of the building are given by \{$(32,26),(26,32),(30,38),(38,32)$\}.
We evaluate the proposed method under different sampling ratios $r_{s}=0.04-0.1$
for reconstructing $\bm{\mathcal{Y}}$. The NMSE results are presented
in Fig.~\ref{fig:direct_block}. It can be observed that the proposed
method consistently outperforms baselines KNN and TPS, achieving more
than $40$\% improvement in reconstruction accuracy across all sampling
ratios. The superiority of the proposed method is mainly due to its
structured tensor formulation: it enforces near-Toeplitz angular patterns
and monotone distance-dependent attenuation through $\{\bm{G}_{r}\}$
and $\{\bm{\rho}_{r}\}$, which allows the model to reconstruct both
the LOS region and the shadowed region with few samples while preserving
the sharp power drop at the blockage. TPS, with its global smoothness
prior, oversmooths this drop, and KNN suffers from unreliable local
averaging near sparsely sampled blocked areas. 
\begin{figure}
\includegraphics{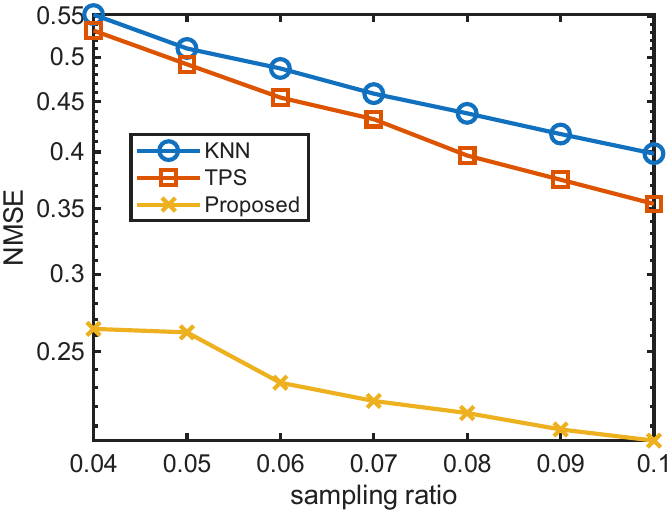}\caption{\label{fig:direct_block}Reconstruction NMSE of $\bm{\mathcal{Y}}$
under the coexistence of LOS and obstruction. The proposed method
consistently achieves lower NMSE than the KNN and TPS baselines across
all sampling ratios, demonstrating a clear advantage in handling the
mixed propagation scenario.}
\end{figure}
\begin{figure}
\includegraphics{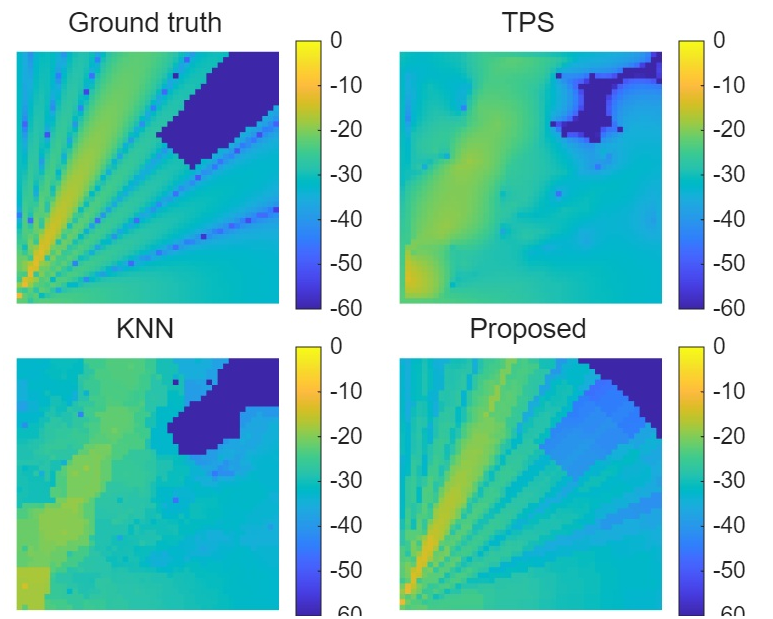}\caption{\label{fig:direct_reflect}Visual comparison of reconstructed MIMO
beam map (LOS + obstruction) at sampling ratio $r_{s}=0.1$. The proposed
method better preserves the overall spatial pattern and partially
recovers the blocked region, whereas TPS and KNN produce distorted
or over-smoothed reconstructions.}
\end{figure}

To further illustrate the reconstruction quality, we set the sampling
ratio to $r_{s}=0.1$ and present the visual comparisons in Fig.~\ref{fig:direct_reflect}.
The proposed method better preserves the spatial pattern and partially
restores the blocked region, whereas TPS and KNN produce distorted
results. The corner region remains difficult to recover due to the
absence of measurements.

\subsection{Simulation for Coexistence of LOS, Reflection and Obstruction}

In this subsection, we examine the reconstruction performance when
both direct and reflected paths exist with obstruction, as illustrated
in Fig. \ref{fig:Visible-plot-of_G12}. To model the reflected path,
a mirror building is generated by mirroring the real building with
respect to the wall, with coordinates \{$(32,70),(26,64),(30,58),(38,64)$\}.
\begin{figure}
\includegraphics{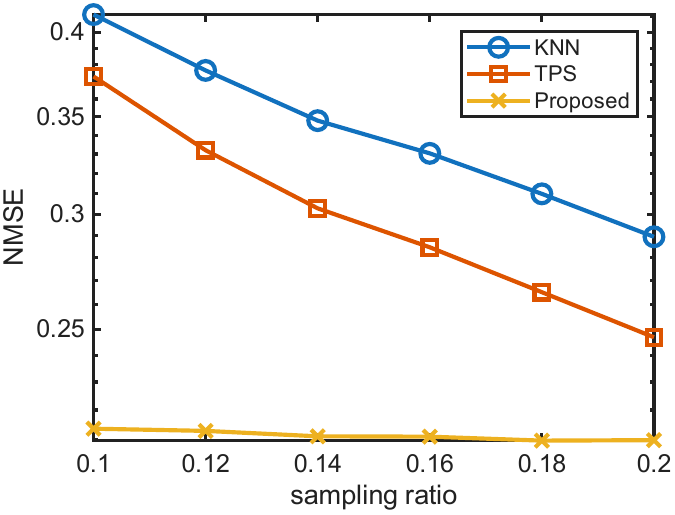}\caption{\label{fig:direct_reflect_block}Reconstruction NMSE of $\bm{\mathcal{Y}}$
under the coexistence of LOS, reflection and obstruction. The proposed
method consistently achieves lower NMSE than the KNN and TPS baselines
across all sampling ratios, demonstrating a clear advantage in handling
the mixed propagation scenario.}
\end{figure}
\begin{figure}
\includegraphics{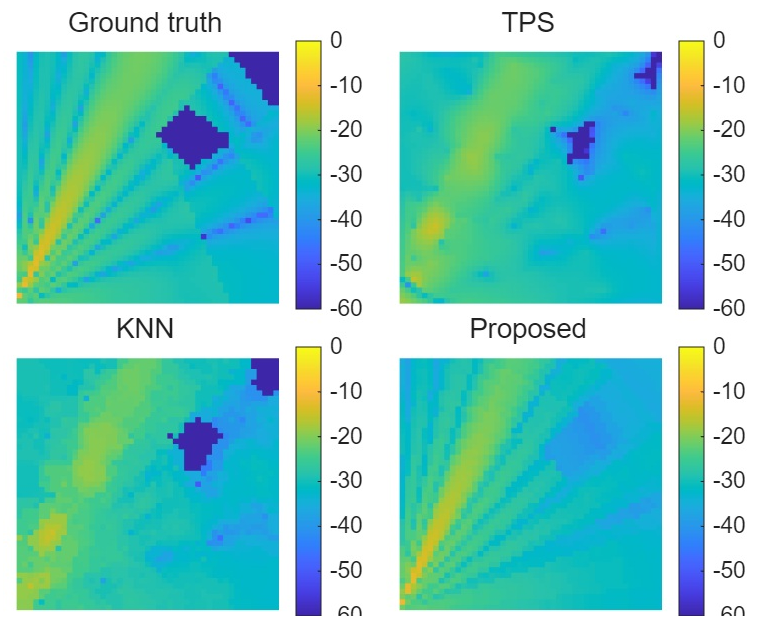}\caption{\label{fig:direct_reflect-blcok_visible}Visual comparison of reconstructed
MIMO beam map (LOS + reflection + obstruction) at sampling ratio $\rho=0.1$.
The proposed method better preserves the overall spatial pattern and
partially recovers the blocked region, whereas TPS and KNN produce
distorted or over-smoothed reconstructions.}
\end{figure}

We evaluate the proposed method under different sampling ratios $r_{s}=0.1-0.2$
for reconstructing $\bm{\mathcal{Y}}$. The NMSE results are presented
in Fig.~\ref{fig:direct_reflect_block}. It can be observed that
the proposed method consistently outperforms baselines KNN and TPS,
achieving more than $20$\% improvement in reconstruction accuracy
under low sampling ratios.

To further illustrate the reconstruction quality, we set the sampling
ratio to $r_{s}=0.1$ and present the visual comparisons in Fig.~\ref{fig:direct_reflect-blcok_visible}.
The proposed method better preserves the spatial pattern and partially
restores the blocked region, whereas TPS and KNN produce distorted
results.

\section{Conclusion\label{sec:Conclusion}}

This paper proposed a structure-aware approach for reconstructing
MIMO beam map from sparse measurements. By transforming Cartesian
observations into a polar coordinate system, we revealed a matrix
and vector outer-product structure associated with different propagation
conditions. The matrix representing beam-space gain was shown to have
an intrinsic Toeplitz pattern, while the vector captured distance-dependent
attenuation. Leveraging these properties, we developed a Toeplitz-structured
matrix-vector tensor decomposition framework and an alternating minimization
algorithm for efficient reconstruction. Simulation results confirmed
that the proposed method achieves more than $20$\% improvement in
reconstruction NMSE across a wide range of sampling ratios, outperforming
interpolation methods and conventional \ac{btd} models. These results
highlight the effectiveness of leveraging the latent angular and radial
structure embedded in MIMO beam map.

% ----

\appendices{}

% ----

\section{Proof of Lemma \ref{lem:symemtric}\label{sec:Proof-of-Lemma symmetric}}

Since $\text{sin}(\theta_{i})=i/I$, $\text{sin}(\phi_{j})=j/J$,
and $I=J$, the difference between the sine terms depends solely on
the index difference. Specifically, for any integer shift $\delta$
such that indices $i+\delta$ and $j+\delta$ are valid, we have
\begin{align}
 & \text{sin}(\theta_{i+\delta})-\text{sin}(\phi_{j+\delta})\nonumber \\
= & \frac{i+\delta}{I}-\frac{j+\delta}{I}\nonumber \\
= & \frac{i-j}{I}=\text{sin}(\theta_{i})-\text{sin}(\phi_{j}).\label{eq:toeplitz condition}
\end{align}

Then, under (\ref{eq:toeplitz condition}), we examine the difference
between diagonal-shifted entries
\begin{align*}
 & [\bm{G}]_{i,j}-[\bm{G}]_{i+\delta,j+\delta}\\
 & =\Bigg(\left|\sum_{n=0}^{N_{t}-1}\text{exp}^{j\pi n(\text{sin}(\theta_{i})-\text{sin}(\phi_{j}))}\right|^{2}\\
 & \qquad\qquad-\left|\sum_{n=0}^{N_{t}-1}\text{exp}^{j\pi n(\text{sin}(\theta_{i+\delta})-\text{sin}(\phi_{j+\delta}))}\right|^{2}\Bigg)\\
 & =0.
\end{align*}
This implies that matrix $\bm{G}$ has a Toeplitz structure.

For the symmetric property, we compare $[\bm{G}]_{i,j}$ and $[\bm{G}]_{j,i}$.
Note that
\begin{equation}
\sin(\theta_{j})-\sin(\phi_{i})=\frac{j-i}{I}=-(\sin(\theta_{i})-\sin(\phi_{j})).
\end{equation}

Thus, the entry $[\bm{G}]_{j,i}$ can be written as
\begin{align}
[\bm{G}]_{j,i} & =\left|\sum_{n=0}^{N_{t}-1}\exp\left(j\pi n(\sin(\theta_{j})-\sin(\phi_{i}))\right)\right|^{2}\nonumber \\
 & =\left|\sum_{n=0}^{N_{t}-1}\exp\left(-j\pi n(\sin(\theta_{i})-\sin(\phi_{j}))\right)\right|^{2}.
\end{align}

Let $z=\sum_{n=0}^{N_{t}-1}\exp\left(j\pi n(\sin(\theta_{i})-\sin(\phi_{j}))\right)$.
The expression for $[\bm{G}]_{j,i}$ corresponds to $|\bar{z}|^{2}$
(where $\bar{z}$ is the complex conjugate of $z$). Using the property
that $|z|=|\bar{z}|$ for any complex number, we have
\begin{equation}
[\bm{G}]_{i,j}-[\bm{G}]_{j,i}=|z|^{2}-|\bar{z}|^{2}=0.
\end{equation}

This concludes that $\bm{G}$ is symmetric.

\bibliographystyle{IEEEtran}
\bibliography{IEEEabrv,StringDefinitions,ChenBibCV,JCgroup}

\end{document}